\documentclass[12pt]{article}
\usepackage[T2A]{fontenc}
\usepackage{latexsym,amsthm,amssymb,amsfonts,amsmath,amscd,wrapfig}
\usepackage{graphicx}
\usepackage{graphicx}
\usepackage{citehack}

\makeatletter \@addtoreset{equation}{section} \makeatother

\newtheorem{proposition}{Proposition}
\newtheorem{theorem}{Theorem}
\newtheorem{lemma}{Lemma}
\newtheorem{remark}{Remark}

\def\dfrac{\displaystyle\frac}
\def\sumd{\displaystyle\sum}

\def\intd{\displaystyle\int}
\def\prodd{\displaystyle\prod}
\def\Tr{\hbox{tr}\,}
\def\mdet{\mathrm{det}}
\begin{document}

\title{On the correlation functions of the characteristic polynomials of the hermitian sample covariance ensemble}
\author{ T. Shcherbina\\
 Institute for Low Temperature Physics, Kharkov,
Ukraine. \\E-mail: t\underline{ }shcherbina@rambler.ru
}
\date{}
\date{}
\maketitle
\begin{abstract}
We consider asymptotic behavior of the correlation functions of the characteristic polynomials of the
hermitian sample covariance matrices $H_n=n^{-1}A_{m,n}^*A_{m,n}$, where $A_{m,n}$ is a $m\times n$ complex matrix with
independent and identically distributed entries $\Re a_{\alpha j}$
and $\Im a_{\alpha j}$. We show that for the correlation function of any even order the
asymptotic behavior in the bulk and at the edge of the spectrum coincides with those for the Gaussian
Unitary Ensemble up to a factor, depending only on the fourth moment of the common probability law of entries
$\Re a_{\alpha j}$, $\Im a_{\alpha j}$, i.e. the higher moments do not
contribute to the above limit.
\end{abstract}
\section{Introduction}
Characteristic polynomials of random matrices have been actively studied
in the last years. The interest was initially stimulated by the similarity between the
asymptotic behavior of the moments of characteristic polynomials of a random matrix from the Circular Unitary Ensemble
and the moments of the Riemann $\zeta$-function along its critical line (see \cite{K-Sn:00}).
But with the emerging connections to the quantum chaos, integrable systems, combinatorics, representation
theory and others, it has become apparent that the characteristic polynomials of random matrices are also
of independent interest. This motivates the studies of the moments of characteristic
polynomials for other random matrix ensembles (see e.g. \cite{HKC:01}, \cite{Me-Nor:01}, \cite{Br-Hi:01},
\cite{St-Fy:03_1}, \cite{BaD:03}, \cite{St-Fy:03}, \cite{Va:03}, \cite{BorSt:06}, \cite{Got-K:08}, \cite{TSh:11}).

In this paper we consider the hermitian sample covariance ensembles with symmetric entries distributions,
i.e. $n\times n$ random matrices of the form
\begin{equation}\label{H}
H_n=n^{-1}A_{m,n}^*A_{m,n},
\end{equation}
where $A_{m,n}$ is an $m\times n$ complex matrix with independent and identically distributed entries $\Re a_{\alpha j}$
and $\Im a_{\alpha j}$ such that
\begin{equation}\label{W}
\begin{array}{c}
\mathbf{E}\{a_{\alpha j}\}=\mathbf{E}\{(a_{\alpha j})^2\}=0,\quad \mathbf{E}\{|a_{\alpha j}|^2\}=1,\quad \alpha=1,..,m,\,j=1,..,n,\\
\mathbf{E}\{\Re^{2l+1}a_{\alpha j}\}=\mathbf{E}\{\Im^{2l+1} a_{\alpha j}\}=0, \quad l\in \mathbb{N}.
\end{array}
\end{equation}
We assume that $m$ belongs to a sequence $\{m_n\}_{n=1}^\infty$ such that
\begin{equation}\label{c}
c_{m,n}:=\dfrac{m_n}{n}\to c\ge 1,\quad n\to\infty.
\end{equation}
We below denote this limit as "$\lim\limits_{m,n\to\infty}\ldots$".

Let $\lambda_1^{(n)},\ldots,\lambda_n^{(n)}$ be the eigenvalues of
$H_n$. Define their Normalized Counting Measure
(NCM) as
\begin{equation}  \label{NCM}
N_n(\triangle)=\sharp\{\lambda_j^{(n)}\in
\triangle,j=1,..,n \}/n,\quad N_n(\mathbb{R})=1,
\end{equation}
where $\triangle$ is an arbitrary interval of the real axis.
The behavior of $N_n$, as $n\to\infty$, is studied well enough.
In particular, it was shown
in \cite{Mar-Pa:67} that $N_{n}$ converges weakly in probability to a non-random measure
$N$ which is called the limiting NCM of the ensemble. The measure $N$ is absolutely continuous
and its density $\rho$ is given by the well-known Marchenko-Pastur
law:
\begin{equation}\label{rho_mp}
\rho(\lambda)=\left\{
\begin{array}{ll}
\dfrac{1}{2\pi \lambda}\sqrt{(\lambda_+-\lambda)
(\lambda-\lambda_-)},& \lambda\in \sigma,\\
0,&\lambda\not\in \sigma,
\end{array}\right.
\end{equation}
where
\begin{equation}\label{lam_pm}
\lambda_\pm=(1\pm\sqrt{c})^2, \quad \sigma=((1-\sqrt{c})^2,(1+\sqrt{c})^2).
\end{equation}
The mixed moments (or the correlation functions) of characteristic polynomials are
\begin{equation}\label{F}
F_{2k}(\Lambda)=\displaystyle\int\limits_{\mathcal{H}_n^+}\prod\limits_{j=1}^{2k}\mdet(\lambda_j-H_n)P_n(d\,H_n),
\end{equation}
where $\mathcal{H}_n^+$ is the space of positive definite hermitian $n\times n$ matrices,
$P_n(d\,H_n)$ is a probability law of the
$n\times n$ random matrix $H_n$,
and $\Lambda=\{\lambda_j\}_{j=1}^{2k}$ are real or complex parameters
that may depend on $n$.

We are interested in the asymptotic behavior of (\ref{F}) for matrices (\ref{H}) as $m,n\to\infty$ and
for $j=1,..,2k$
\begin{equation}\label{lam}
\lambda_j=\left\{
\begin{array}{ll}
\lambda_0+\xi_j/n\rho(\lambda_0),& \lambda_0\in \sigma,\\
\lambda_0+\xi_j/(n\gamma_{\pm})^{2/3},& \lambda_0=\lambda_\pm,
\end{array}\right.
\end{equation}
where $\lambda_\pm$ and $\sigma$ are defined in (\ref{lam_pm}),
\begin{equation}\label{gam}
\gamma_{\pm}=\dfrac{c^{1/4}}{(1\pm\sqrt{c})^2},
\end{equation}
$\rho$ is defined in (\ref{rho_mp}), and
$\widehat{\xi}=\{\xi_j\}_{j=1}^{2k}$ are real parameters varying in $[-M,M]\subset \mathbb{R}$.


In the case of hermitian matrix models the asymptotic behavior of (\ref{F}) was obtained by using the method of
orthogonal polynomials (see \cite{Br-Hi:00,St-Fy:03}). Unfortunately, the method of orthogonal polynomials
can not be applied to the general case of the
hermitian sample covariance ensembles (\ref{H}) -- (\ref{W}). In the paper \cite{TSh:11} the method based on
the Grassmann integration was developed to study the asymptotic behavior of the correlation functions
of any any even number of the characteristic polynomials of the hermitian Wigner ensemble. Here we apply this method
to the hermitian sample covariance ensembles (\ref{H}) -- (\ref{W}).

In \cite{Kos:09} Kosters use the exponential generating function to study the second moment, i.e.
the case $k=1$ in (\ref{F}). It was shown that for $\lambda_0\in\sigma$
\begin{multline*}
\dfrac{1}{n\rho(\lambda_0)}F_2\left(\lambda_0+\xi_1/(n\rho(\lambda_0)),
\lambda_0+\xi_2/(n\rho(\lambda_0))\right)=2\pi\lambda_0^{n-m}c_{m,n}^{m+1/2}\\
\times
e^{-n-m}\exp\{n\lambda_0+\alpha(\lambda_0)(\xi_1+\xi_2)+2\kappa_4\}\dfrac{\sin(\pi(\xi_1-\xi_2))}{\pi(\xi_1-\xi_2)}(1+o(1)),
\end{multline*}
where
\begin{equation}\label{alpha,kap}
\alpha(\lambda_0)=\left\{\begin{array}{ll}
\dfrac{\lambda_0-c+1}{2\lambda_0\rho(\lambda_0)},&\lambda_0\in\sigma,\\
(1\pm\sqrt{c})^{-1}\gamma_\pm^{-2/3},&\lambda_0=\lambda_\pm,
\end{array}\right.
\quad \kappa_4=\mu_4-3/4,
\end{equation}
$\gamma_\pm$ is defined in (\ref{gam}), and $\mu_4$ is the fourth moment of the probability law of $\Im W_{jk}$, $\Re W_{jk}$.
In \cite{Kos:09} for the case $c>1$, $m=cn+o(n^{1/3})$, $k=1$ the asymptotic behavior at the edge of the spectrum
(i.e. for $\lambda_0=\lambda_\pm$) was also obtained:
\begin{multline*}
\dfrac{1}{(n\gamma_\pm)^{2/3}}F_2\left(\lambda_0+\xi_1/(n\gamma_\pm)^{2/3},
\lambda_0+\xi_2/(n\gamma_\pm)^{2/3}\right)=2\pi (1\pm\sqrt{c})^{2(n-m)}c^{m+1/2}\\
\times
e^{2n\sqrt{c}}\exp\{n^{1/3}\alpha(\lambda_\pm)(\xi_1+\xi_2)+2\kappa_4\}A(\xi_1,\xi_2)(1+o(1))
\end{multline*}
with
\begin{equation}  \label{A}
A(x,y)=\displaystyle\frac{\mathrm{Ai}'(x)\mathrm{Ai}(y)-\mathrm{Ai}(x)\mathrm{Ai}'(y)}{x-y},
\end{equation}
where $\mathrm{Ai}(x)$ is the Airy function
\begin{equation}\label{Ai}
\mathrm{Ai}(x)=\dfrac{1}{2\pi}\int\limits_Se^{is^3/3+isx}d\,s,
\end{equation}
$$S=\{z\in \mathbb{C}|\arg z=\pi/6\,\,\hbox{or}\,\arg z=5\pi/6\}.$$
In this paper we consider the general case $k\ge 1$ of (\ref{F}) for the random matrices (\ref{H}).
Define
\begin{equation}\label{D_xi}
D^{(n)}(\xi,\lambda_0)=\left\{
\begin{array}{ll}
(n\rho(\lambda_0))^{-1}F_{2}\Big(\lambda_0+\xi/(n\rho(\lambda_0)),
\lambda_0+\xi/(n\rho(\lambda_0))\Big),& \lambda_0\in\sigma,\\
(n\gamma_\pm)^{-2/3}F_{2}\Big(\lambda_0+\xi/(n\gamma_\pm)^{2/3},
\lambda_0+\xi/(n\gamma_\pm)^{2/3}\Big),& \lambda_0=\lambda_\pm,
\end{array}\right.
\end{equation}
and denote
\begin{equation}\label{D_2k}
D_{2k}(\lambda_0)=\prod\limits_{l=1}^{2k}\sqrt{D^{(n)}(\xi_l,\lambda_0)}.
\end{equation}

  The main results of the paper are the following two theorems:
\begin{theorem}\label{thm:1}
Let the entries $\Im a_{\alpha j}$, $\Re a_{\alpha j}$ of the matrices (\ref{H}) have a
symmetric probability distribution with finite first $4k$ moments. Then we have for $k\ge 1$
\begin{multline}\label{lim1}
\lim\limits_{n\to\infty}\dfrac{1}{(n\rho(\lambda_0))^{k^2}D_{2k}(\lambda_0)}
F_{2k}\left(\Lambda_0+\widehat{\xi}/(n\rho(\lambda_0))\right)\\
=\dfrac{c^{k(k-1)/2}\exp\{k(k-1)\kappa_4(c-\lambda_0+1)^2c^{-1}\}}{\Delta(\xi_1,...,\xi_m)
\Delta(\xi_{k+1},...,\xi_{2k})}\mdet
\left\{\dfrac{\sin(\pi(\xi_i-\xi_{k+j}))}{\pi(\xi_i-\xi_{k+j})}
\right\}_{i,j=1}^k,
\end{multline}
where $F_{2k}$ and $\rho(\lambda)$ are defined in (\ref{F}) and (\ref{rho_mp}),
$\Lambda_0=(\lambda_0,\ldots,\lambda_0)\in \mathbb{R}^{2k}$,
$\lambda_0\in\sigma$, $\widehat{\xi}=\{\xi_j\}_{j=1}^{2k}$, and
$\kappa_4$ and $\sigma$ are defined in (\ref{alpha,kap}) and (\ref{lam_pm}).
\end{theorem}
\begin{theorem}\label{thm:2}
Let the entries $\Im a_{\alpha j}$, $\Re a_{\alpha j}$ of the matrices (\ref{H}) have a
symmetric probability distribution with finite first $4k$ moments, $\lambda=\lambda_\pm$ and
let $m$ belong to a sequence $\{m_n\}_{n=1}^\infty$ such that
\begin{equation}\label{cond_m}
m_n=c\,n+n^{1/3}\varepsilon_n,\quad c>1,
\end{equation}
where $\varepsilon_n\to 0, \,n\to\infty$.
Then we have for $k\ge 1$
\begin{multline*}
\lim\limits_{n\to\infty}\dfrac{1}{(n\gamma_\pm)^{2k^2/3}D_{2k}(\lambda_\pm)}
F_{2k}\left(\Lambda_0+\widehat{\xi}/(n\gamma_\pm)^{2/3}\right)\\
=\dfrac{c^{k(k-1)/2}\exp\{4k(k-1)\kappa_4\}}{\Delta(\xi_1,...,\xi_m)
\Delta(\xi_{k+1},...,\xi_{2k})}\mdet
\Big\{\mathrm{A}(\xi_j,\xi_{k+l})
\Big\}_{i,j=1}^k,
\end{multline*}
where $F_{2k}$ and $\gamma_\pm$ are defined in (\ref{F}) and (\ref{gam}),
$\Lambda_0=(\lambda_\pm,\ldots,\lambda_\pm)\in \mathbb{R}^{2k}$,
$\widehat{\xi}=\{\xi_j\}_{j=1}^{2k}$, and $\kappa_4$ and $\lambda_\pm$ are defined in (\ref{alpha,kap}) and (\ref{lam_pm}).
\end{theorem}

The theorems show that the above limits for the mixed moments of the characteristic polynomials
for random matrices (\ref{H}) coincide with those for the Gaussian Unitary Ensemble up to a
factor depending only on the fourth moment of the common probability law of the entries
$a_{\alpha j}$, i.e., that the higher moments of the law do not
contribute to the above limit. This is a manifestation of the universality,
that can be compared with the universality of the local bulk regime for Wigner matrices (see \cite{Er:10} and
references therein).

The paper is organized as follows. In Section $2$ we obtain a convenient asymptotic integral representation for
$F_{2k}$,
using the integration
over the Grassmann variables and the Harish Chandra/Itzykson-Zuber formula for integrals
over the unitary group. The method is similar to that of \cite{TSh:11}. In Section $3$ and $4$ we prove Theorem
\ref{thm:1} and \ref{thm:2}, applying the steepest descent method
to the integral representation.

We denote by $C, C_1$, etc. various $n$-independent quantities below, which
can be different in different formulas.

\section{The integral representation.}
In this section we obtain the integral representation for the mixed moments $F_{2k}$ (\ref{F})
of the characteristic polynomials, i.e. we prove the following proposition
\begin{proposition}\label{p:int_repr}
Let $\Lambda_{2k}=\Lambda_0+\widehat{\xi}/(an)^{\alpha}$, where
$\Lambda_0=\mathrm{diag}\{\lambda_0,..,\lambda_0\}$, $\widehat{\xi}=\mathrm{diag}\{\xi_1,..,\xi_{2k}\}$,
and
\begin{eqnarray}\label{a}
a&=&\left\{
\begin{array}{ll}
\rho(\lambda_0),&\lambda_0\in \sigma,\\
\gamma_\pm,& \lambda_0=\lambda_\pm,
\end{array}\right.\\ \label{bet}
\beta&=&\left\{
\begin{array}{ll}
1,&\lambda_0\in \sigma,\\
2/3,& \lambda_0=\lambda_\pm,
\end{array}\right.
\end{eqnarray}
where $\sigma$ and $\lambda_\pm$ are defined in (\ref{lam_pm}), and let $F_{2k}(\Lambda_{2k})$ of (\ref{F})
be the correlation function of the characteristic polynomials. Then we have for every $k$
\begin{equation}\label{F_int}
\begin{array}{c}
D_{2k}^{-1}(\lambda_0)F_{2k}(\Lambda_{2k})=\dfrac{n^{2k^2}(n^{\beta-1}a^{\beta})^{k(2k-1)}}{2^k\pi^ke^{2kn}D_{2k}(\lambda_0)}
\displaystyle\oint\limits_{\omega}\prod\limits_{j=1}^{2k}d\,v_j
e^{\sum\limits_{j=1}^{2k}(n\lambda_0v_j+n^{1-\beta}a^{-\beta}\xi_jv_j)}\dfrac{\Delta(V)}{\Delta(\widehat{\xi})}\\
\prod\limits_{l=1}^{2k}\dfrac{(1-v_l)^{m}}{v_l^{n+2k}}\exp\Big\{2c_{m,n}\kappa_4S_2((I-V)\Lambda_0)
\prod\limits_{l=1}^{2k}\frac{v_l}{1-v_l}\Big\}(1+O(n^{-\beta})),\,\, m,n\to\infty,
\end{array}
\end{equation}
where $V=\mathrm{diag}\{v_1,\ldots,v_{2k}\}$,
\begin{equation}\label{S}
S_2(A)=\dfrac{1}{2}\dfrac{d^2}{dx^2}\mdet(x - A)\bigg|_{x=c},
\end{equation}
$D_{2k}(\lambda_0)$ is defined in (\ref{D_2k}) and $\omega$ is any closed contour encircling $0$.
\end{proposition}

To this end we use the integration
over the Grassmann variables. The integration was introduced by Berezin
and widely used in the physics literature (see e.g. \cite{Ber} and \cite{Ef}).
For the reader convenience we give a brief outline of the techniques here.

\subsection{Grassmann integration}
Let us consider two sets of formal variables
$\{\psi_j\}_{j=1}^n,\{\overline{\psi}_j\}_{j=1}^n$, which satisfy the anticommutation
conditions
\begin{equation}\label{anticom}
\psi_j\psi_k+\psi_k\psi_j=\overline{\psi}_j\psi_k+\psi_k\overline{\psi}_j=\overline{\psi}_j\overline{\psi}_k+
\overline{\psi}_k\overline{\psi}_j=0,\quad j,k=1,..,n.
\end{equation}
These two sets of variables $\{\psi_j\}_{j=1}^n$ and
$\{\overline{\psi}_j\}_{j=1}^n$ generate the Grassmann algebra
$\mathfrak{A}$. Taking into account that $\psi_j^2=0$, we have that all elements
of $\mathfrak{A}$ are polynomials of $\{\psi_j\}_{j=1}^n$ and
$\{\overline{\psi}_j\}_{j=1}^n$. We can also define functions of the Grassmann variables. Let $\chi$ be
an element of $\mathfrak{A}$. For any analytical function $f$ we mean by
$f(\chi)$ the element of $\mathfrak{A}$ obtained by
substituting $\chi$ in the Taylor series of $f$. Since
$\chi$ is a polynomial of $\{\psi_j\}_{j=1}^n$, $\{\overline{\psi}_j\}_{j=1}^n$,
there exists such $l$ that $\chi^l=0$ and hence the series
terminates after a finite number of terms and so $f(\chi)\in
\mathfrak{A}$.

Note also that if $\chi$ is the sum of the products of even numbers of the Grassmann variables, then,
according to the definition
of the functions of the Grassmann variables, expanding $(z-\chi)^{-1}$ into the series we obtain for any analytic function $f$
\begin{equation}\label{cont_int_gr}
\displaystyle\oint\limits_{\Omega}\dfrac{f(z)}{z-\chi}\dfrac{dz}{2\pi i}=f(\chi),
\end{equation}
where $\Omega$ is any closed contour encircling $0$.

Following Berezin \cite{Ber}, we define the operation of
integration with respect to the anticommuting variables in a formal
way:
\begin{equation}\label{int_gr}
\intd d\,\psi_j=\intd d\,\overline{\psi}_j=0,\quad \intd
\psi_jd\,\psi_j=\intd \overline{\psi}_jd\,\overline{\psi}_j=1.
\end{equation}
This definition can be extended on the general element of $\mathfrak{A}$ by
the linearity. A multiple integral is defined to be a repeated
integral. The "differentials" $d\,\psi_j$ and
$d\,\overline{\psi}_k$ anticommute with each other and with the
variables $\psi_j$ and $\overline{\psi}_k$.

Thus, if
$$
f(\chi_1,\ldots,\chi_k)=a_0+\sum\limits_{j_1=1}^k
a_{j_1}\chi_{j_1}+\sum\limits_{j_1<j_2}a_{j_1j_2}\chi_{j_1}\chi_{j_2}+
\ldots+a_{1,2,\ldots,k}\chi_1\ldots\chi_k,
$$
then
\[
\intd f(\chi_1,\ldots,\chi_k)d\,\chi_k\ldots d\,\chi_1=a_{1,2,\ldots,k}.
\]

   Let $A$ be an ordinary hermitian matrix. The following Gaussian integral is well-known
\begin{equation}\label{G_C}
\intd \exp\Big\{-\sum\limits_{j,k=1}^nA_{j,k}z_j\overline{z}_k\Big\}
\prod\limits_{j=1}^n\dfrac{d\,\Re z_jd\,\Im z_j}{\pi}=\dfrac{1}{\mdet A}.
\end{equation}
One of the important formulas of the Grassmann variables theory is
the analog of this formula for the Grassmann algebra (see \cite{Ber}):
\begin{equation}\label{G_Gr}
\int \exp\Big\{\sum\limits_{j,k=1}^nA_{j,k}\overline{\psi}_j\psi_k\Big\}
\prod\limits_{j=1}^nd\,\overline{\psi}_jd\,\psi_j=\mdet A.
\end{equation}
Besides, we have
\begin{equation}\label{G_Gr_1}
\int \prod\limits_{p=1}^q\overline{\psi}_{l_p}\psi_{s_p}\exp\Big\{\sum\limits_{j,k=1}^nA_{j,k}\overline{\psi}_j
\psi_k\Big\}\prod\limits_{j=1}^nd\,\overline{\psi}_jd\,\psi_j=\mdet A_{l_1,..,l_q;s_1,..,s_q},
\end{equation}
where $A_{l_1,..,l_q;s_1,..,s_q}$ is a $(n-q)\times (n-q)$ minor of the matrix $A$ without lines $l_1,..,l_q$ and
columns $s_1,..,s_q$.

\subsection{Asymptotic integral representation for $F_2$}
In this subsection we obtain (\ref{F_int}) for $k=1$ by
using the Grassmann integrals.
This formula was obtained in \cite{Kos:09} by using another method. We give here a detailed proof to
show the basic ingredients of our techniques that will be elaborated in the next subsection
to obtain the asymptotic integral representation of (\ref{F}) for $k>1$.

Using (\ref{G_Gr}), we obtain from (\ref{F})
\begin{equation}\label{ne_usr_2}
\begin{array}{c}
F_2(\Lambda_2)={\bf
E}\bigg\{\displaystyle\int e^{\sum\limits_{r=1}^2\sum\limits_{p,q=1}^n(\lambda_l-H)_{p,q}
\overline{\psi}_{pr}\psi_{qr}}d\,\Psi_{2,n}\bigg\}\\
={\bf E}\bigg\{\displaystyle\int \prod\limits_{\alpha=1}^m
e^{-\frac{1}{n}\sum\limits_{r=1}^2\big(\sum\limits_{p=1}^n\overline{a}_{\alpha p}\overline{\psi}_{pr}\big)
\big(\sum\limits_{q=1}^na_{\alpha q}\psi_{qr}\big)}
 e^{\sum\limits_{s=1}^2\lambda_s\sum\limits_{p=1}^n
\overline{\psi}_{ps}\psi_{ps}}d\,\Psi_{2,n}\bigg\}\\
={\bf E}\bigg\{\displaystyle\int \prod\limits_{\alpha=1}^m\prod\limits_{r=1}^2
\bigg(1-\frac{1}{n}\sum\limits_{p,q=1}^n\overline{a}_{\alpha p}a_{\alpha q}\overline{\psi}_{pr}
\psi_{qr}\bigg)
 e^{\sum\limits_{s=1}^2\lambda_s\sum\limits_{p=1}^n
\overline{\psi}_{ps}}d\,\Psi_{2,n}\bigg\},
\end{array}
\end{equation}
since for any $\alpha=1,..,m$ and any $r=1,2$ we have according to (\ref{anticom})
\begin{equation}\label{kvadr}
\bigg(\sum\limits_{p=1}^n\overline{a}_{\alpha p}\overline{\psi}_{pr}\bigg)^2=
\bigg(\sum\limits_{q=1}^na_{\alpha q}\psi_{qr}\bigg)^2=0.
\end{equation}
Here $\{\psi_{jl}\}_{j,l=1}^{n \, 2}$ are the
Grassmann variables ($n$ variables for each determinant in (\ref{F})) and
\begin{equation}\label{dPsi}
d\,\Psi_{s,l}=\prod\limits_{r=1}^{s}\prod\limits_{j=1}^l
d\,\overline{\psi}_{jr}d\,\psi_{jr}.
\end{equation}
In view of (\ref{W}) and (\ref{anticom}) we get
\begin{eqnarray}\label{usr_a_2}
{\bf E}\bigg\{\prod\limits_{r=1}^2
\bigg(1-\frac{1}{n}\sum\limits_{p,q=1}^n\overline{a}_{\alpha p}a_{\alpha q}\overline{\psi}_{pr}
\psi_{qr}\bigg)\bigg\}
=1-\frac{1}{n}\sum\limits_{r=1}^2\sum\limits_{p,q=1}^n{\bf E}\left\{\overline{a}_{\alpha p}a_{\alpha q}\right\}\overline{\psi}_{pr}
\psi_{qr}\\ \notag
+\frac{1}{n^2}\sum\limits_{p_1,q_1=1}^n\sum\limits_{p_2,q_2=1}^n{\bf E}\left\{\overline{a}_{\alpha p_1}a_{\alpha q_1}
\overline{a}_{\alpha p_2}a_{\alpha q_2}\right\}\overline{\psi}_{p_1 1}\psi_{q_1 1}\overline{\psi}_{p_2 2}\psi_{q_2
2}\\ \notag
=1-\frac{1}{n}\sum\limits_{r=1}^2\sum\limits_{p=1}^n\overline{\psi}_{pr}
\psi_{pr}+\frac{1}{n^2}\sum\limits_{p\ne q}\overline{\psi}_{p 1}\psi_{p 1}\overline{\psi}_{q 2}\psi_{q
2}-\frac{1}{n^2}\sum\limits_{p\ne q}\overline{\psi}_{p 1}\psi_{p 2}\overline{\psi}_{q 2}\psi_{q
1}\\ \notag
+\frac{2\mu_4+1/2}{n^2}\sum\limits_{p=1}^n\overline{\psi}_{p1}\psi_{p1}\overline{\psi}_{p2}\psi_{p2}=
\mdet\, Q_2^{(n)}+\frac{2\kappa_4}{n^2}\sum\limits_{p=1}^n\overline{\psi}_{p1}
\psi_{p1}\overline{\psi}_{p2}\psi_{p2},
\end{eqnarray}
where $\Psi_{s}^{(l)}$ and $Q_s^{(l)}$ are the matrix with Grassmann entries
\begin{equation}\label{Psi_2}
\Psi_{s}^{(l)}=\bigg\{\sum\limits_{p=1}^l\overline{\psi}_{pr}
\psi_{pt}\bigg\}_{r,t=1}^s,\quad Q_s^{(l)}=1-n^{-1}\Psi_{s}^{(l)},
\end{equation}
$\mu_4$ is the $4$-th moment of the probability law of $\Im a_{\alpha j}$, $\Re a_{\alpha j}$ of (\ref{W}),
and $\kappa_4$ is defined in (\ref{alpha,kap}).

Thus, (\ref{ne_usr_2}) and (\ref{usr_a_2}) yield
\begin{equation}\label{usr_2}
\begin{array}{c}
F_2(\Lambda_2)
=\displaystyle
\int e^{\Tr\Psi_{2}^{(n)}\Lambda_2}\bigg(\mdet\, Q_2^{(n)}+\frac{2\kappa_4}{n^2}\sum\limits_{p=1}^n\overline{\psi}_{p1}
\psi_{p1}\overline{\psi}_{p2}\psi_{p2}\bigg)^m
d\,\Psi_{2,n}\\
=\sumd\limits_{q=1}^{m}\binom{m}{q}\dfrac{(2\kappa_4)^q}{n^{2q}}\displaystyle
\int e^{\Tr\Psi_{2}^{(n)}\Lambda_2}\mdet^{m-q} Q_2^{(n)}\bigg(\sum\limits_{p=1}^n\overline{\psi}_{p1}
\psi_{p1}\overline{\psi}_{p2}\psi_{p2}\bigg)^qd\,\Psi_{2,n}\\
=\sumd\limits_{q=1}^{m}\binom{m}{q}\dfrac{n!}{(n-q)!}\dfrac{(2\kappa_4)^q}{n^{2q}}\displaystyle
\int e^{\Tr\Psi_{2}^{(n-q)}\Lambda_2}\,\mdet^{m-q} Q_2^{(n-q)}d\,\Psi_{2,n-q}\\
=:\sumd\limits_{q=1}^{m}\binom{m}{q}\dfrac{n!}{(n-q)!}\dfrac{(2\kappa_4)^q}{n^{2q}}I_{2,q}.
\end{array}
\end{equation}
Here we used the symmetry of $\overline{\psi}_{lp}, \psi_{lp}$ and
\[
\intd \overline{\psi}_{p1}\psi_{p1}\overline{\psi}_{p2}\psi_{p2}\,f(\overline{\psi}_{p1},
\psi_{p1},\overline{\psi}_{p2},\psi_{p2})d\,\overline{\psi}_{p1}d\,\psi_{p1}d\,\overline{\psi}_{p2}d\,\psi_{p2}=f(0,0,0,0).
\]
To compute $I_{2,q}$ we use the following lemma
\begin{lemma}\label{l:repr_det}
Let $A$ be any $p\times p$ matrix and let $l$ be a positive integer. Then we have
\begin{equation}\label{repr_det}
\mdet^l
A=K_{p,l}\int\dfrac{e^{\Tr\,AU}}{\mdet^{p+l}U}
d\,\mu(U),
\end{equation}
where
\begin{equation}\label{K_p,l}
K_{p,l}=(-1)^{p(p-1)/2}S_p^{-1}\prod\limits_{s=0}^{p-1}(l+s)!, \quad S_p=\prod\limits_{s=1}^ps!,
\end{equation}
$U$ is a unitary matrix with eigenvalues $\{u_j\}_{j=1}^p$, $W$ is a matrix which diagonalizes $U$ and
\begin{equation}\label{dU}
d\,\mu(U)=\triangle^2(u_1,\ldots,u_p)d\,W\prod\limits_{j=1}^p\dfrac{du_j}{2\pi i},
\end{equation}
where $d\,u_j$ means the integration over the circle $\omega=\{z:|z|=1\}$, $d\,W$ is the Haar measure over
the unitary group $U(p)$, and $\triangle(u_1,\ldots,u_p)$ is the Vandermonde determinant of $u_j$-s.
\end{lemma}
\begin{remark}
1. Lemma \ref{l:repr_det} is a particular case of the superbosonization formula which was proved in the
physics paper \cite{SupB:08}. We give below (see Subsection 2.4) a different proof for this simple case.

2. Since both sides of (\ref{repr_det}) are analytic functions of $a_{i,j}$, we can take $A$ with not necessary complex but
also with even Grassmann elements.

3. Combining (\ref{repr_det}) and (\ref{G_Gr}) we get that for any $p\times p$ matrix $A$
\begin{equation}\label{ch_exp}
\int e^{\mathrm{tr}\,A\Psi^{(l)}}d\,\Psi_{p,l}=K_{p,l}\int \dfrac{e^{\mathrm{tr}\,AU}}{\det^{p+l}U}d\,\mu(U),
\end{equation}
where $\Psi^{(l)}=\{\sum_{s=1}^l\psi_{sj}\psi_{sr}\}_{j,r=1}^{p}$ and $d\,\Psi_{p,l}$ is defined in
(\ref{dPsi}).
\end{remark}
Using Lemma \ref{l:repr_det} and Remark 1.3, we obtain
\begin{align}\notag
I_{2,q}&=K_{2,m-q}\displaystyle
\int \dfrac{e^{\Tr\,\Lambda_2\Psi_{2}^{(n-q)}+\Tr W_2-n^{-1}\Tr\,W_2\Psi_{2}^{(n-q)} }}{\mdet^{m-q+2}W_2}
d\,\Psi_{2,n-q}d\,\mu(W_2)\\ \label{I_2,q_1}
&=K_{2,m-q}K_{2,n-q}\displaystyle
\int\dfrac{e^{\Tr W_2+\Tr\,\Lambda_2U_2-n^{-1}\Tr
U_{2}W_2}}{\mdet^{m-q+2}W_2\mdet^{n-q+2}U_2}d\,\mu(U_2)d\,\mu(W_2)
\\ \notag &=K_{2,n-q}\displaystyle
\int\dfrac{e^{\Tr\,\Lambda_2U_2}\mdet^{m-q}(I-n^{-1}U_{2})}{\mdet^{n-q+2}U_2}d\,\mu(U_2),
\end{align}
where $U_2$ and $W_2$ are unitary $2\times 2$ matrices, and $d\,\mu(U_2)$, $d\,\mu(W_2)$ are defined in (\ref{dU}).

Recall that we are interested in $\Lambda_2=\Lambda_{0,2}+\widehat{\xi}_2/(na)^\beta$, where
$\Lambda_{0,2}=\hbox{diag}\{\lambda_0,\lambda_0\}$,
$\widehat{\xi}_2=\hbox{diag}\{\xi_1,\xi_{2}\}$, and $a$, $\beta$ are defined in (\ref{a}), (\ref{bet}).
Substituting (\ref{dU}) in (\ref{I_2,q_1}) and using that functions $\mdet (I-n^{-1}U_{2})$,
$\Tr\,\Lambda_0U_2$,
and $\mdet\,U_2$
are unitary invariant, we obtain from (\ref{I_2,q_1})
\begin{eqnarray}\label{I_2,q_2}
I_{2,q}=K_{2,n-q}\displaystyle
\int\displaystyle
\oint\limits_\omega e^{\Tr\,\Lambda_{0,2}V_2+(na)^{-\beta}\Tr\,\widehat{\xi}_2W^*V_2W}\prod\limits_{r=1}^2\dfrac{
(1-\frac{v_r}{n})^{m-q}}{v_r^{n-q+2}}(v_1-v_2)^2d\,\mu(W)\dfrac{dv_1dv_2}{(2\pi i)^2}\\ \notag
=\frac{K_{2,n-q}}{n^{2(n-q)}}\displaystyle
\int\displaystyle
\oint\limits_\omega e^{\Tr\,\Lambda_{0,2}V_2+(na)^{-\beta}\Tr\,\widehat{\xi}_2W^*V_2W}\prod\limits_{r=1}^2\dfrac{
(1-v_r)^{m-q}}{v_r^{n-q+2}}(v_1-v_2)^2d\,\mu(W)\dfrac{dv_1dv_2}{(2\pi i)^2},
\end{eqnarray}
where $\omega$ is any closed contour encircling $0$.
The integral over the unitary group $U(2)$ can be computed using the
Harish Chandra/Itsykson-Zuber formula (see e.g. \cite{Me:91}, Appendix~5):
\begin{proposition}\label{p:Its-Z}
Let $A$ be the normal $p\times p$ matrix with distinct eigenvalues $\{a_i\}_{i=1}^p$ and
$B=\mathrm{diag}\{b_1,\ldots,b_p\}$. Then for any symmetric function $f(B)$ of $\{b_j\}_{j=1}^p$  we have
\begin{equation}\label{Its-Zub}
\int\limits_{U(p)}\int e^{\mathrm{tr}\, AU^*BU} \triangle^2(B)f(B)d\,Ud\,B\\
=S_p
\int e^{\sum\limits_{j=1}^pa_jb_j}\dfrac{\triangle(B)}{\triangle(A)} f(b_1,\ldots,b_p)
d\,B,
\end{equation}
where $S_p$ is defined in (\ref{K_p,l}), $d\,B=\prod\limits_{j=1}^pd\,b_j$,
$d\,U$ is the Haar measure over the unitary group $U(n)$ and $\triangle(A)$, $\triangle(B)$ are the
Vandermonde determinants of the eigenvalues $\{a_i\}_{i=1}^p$, $\{b_i\}_{i=1}^p$ of $A$ and $B$.
\end{proposition}
This and formula (\ref{I_2,q_2}) yields
\begin{equation}\label{I_2,q}
I_{2,q}=\frac{2 n^{\beta-1}a^{\beta} K_{2,n-q}}{n^{2(n-q)}}\displaystyle
\oint\limits_\omega e^{n\Tr\,\Lambda_{0,2}V_2+n^{1-\beta}a^{-\beta}\Tr\,\widehat{\xi}_2V_2}\prod\limits_{r=1}^2\dfrac{
(1-v_r)^{m-q}}{v_r^{n-q+2}}\dfrac{(v_1-v_2)dv_1dv_2}{(\xi_1-\xi_2)(2\pi i)^2}.
\end{equation}
Hence, since
\[
\dfrac{n!}{(n-q)!}\cdot\dfrac{(n-q+1)!(n-q)!}{n^{2n-q}}=2\pi n^2 e^{-2n} (1+O(1/n))
\]
we get (\ref{F_int}) for $k=1$ from (\ref{usr_2}), (\ref{K_p,l}), and (\ref{I_2,q}).

\subsection{Asymptotic integral representation for $F_{2k}$}

Using (\ref{G_Gr}) and (\ref{kvadr}), we obtain from (\ref{F}) (cf. (\ref{ne_usr_2}))
\begin{equation}\label{ne_usr_k}
\begin{array}{c}
F_{2k}(\Lambda_{2k})={\bf
E}\bigg\{\displaystyle\int e^{\sum\limits_{r=1}^{2k}\sum\limits_{p,q=1}^n(\lambda_l-H)_{p,q}
\overline{\psi}_{pr}\psi_{qr}}d\,\Psi_{2k,n}\bigg\}\\
={\bf E}\bigg\{\displaystyle\int e^{\sum\limits_{s=1}^{2k}\lambda_s\sum\limits_{p=1}^n
\overline{\psi}_{ps}\psi_{ps}}\prod\limits_{\alpha=1}^m
e^{-\frac{1}{n}\sum\limits_{r=1}^{2k}\big(\sum\limits_{p=1}^n\overline{a}_{\alpha p}\overline{\psi}_{pr}\big)
\big(\sum\limits_{q=1}^na_{\alpha q}\psi_{qr}\big)}
d\,\Psi_{2k,n}\bigg\}\\
={\bf E}\bigg\{\displaystyle\int e^{\sum\limits_{s=1}^{2k}\lambda_s\sum\limits_{p=1}^n
\overline{\psi}_{ps}\psi_{ps}}\prod\limits_{\alpha=1}^m\prod\limits_{r=1}^{2k}
\Big(1-\frac{1}{n}\sum\limits_{p,q=1}^n\overline{a}_{\alpha p}a_{\alpha q}\overline{\psi}_{pr}
\psi_{qr}\Big)
d\,\Psi_{2k,n}\bigg\}.
\end{array}
\end{equation}
In view of (\ref{W}) similarly to (\ref{usr_a_2}) we get
\begin{multline}\label{usr_a_k}
{\bf E}\bigg\{\prod\limits_{r=1}^{2k}
\Big(1-\frac{1}{n}\sum\limits_{p,q=1}^n\overline{a}_{\alpha p}a_{\alpha q}\overline{\psi}_{pr}
\psi_{qr}\Big)\bigg\}
=\mdet\,Q^{(n)}_{2k}\\
+\frac{2\kappa_4}{n^2}\sum\limits_{l_1<l_2,s_1<s_2}\mdet (Q^{(n)}_{2k})^{(l_1,l_2;s_1,s_2)}
\sum\limits_{p=1}^n\overline{\psi}_{pl_1}
\psi_{ps_1}\overline{\psi}_{pl_2}\psi_{ps_2} +
n^{-2}\Phi(\Psi),
\end{multline}
where $Q^{(n)}_{2k}$ is defined in (\ref{Psi_2}),
$\,\mdet (Q^{(n)}_{2k})^{(l_1,l_2;s_1,s_2)}$ is $(2k-2)\times(2k-2)$ minor of matrix
$Q^{(n)}_{2k}$ without lines $s_1,s_2$ and columns $l_1,l_2$,
$\kappa_4$ is defined in (\ref{alpha,kap}) and  $\Phi(\Psi)$ is a polynomial of the variables
$\{(n^{-1}\Psi^{(n)}_{2k})_{r,s}\}_{r,s=1}^{2k}$ and
\begin{equation}\label{sig}
n^{-1}\sigma_{\overline{l},\overline{s}}^{(n)}=\dfrac{1}{n}\sum\limits_{p=1}^n\prod\limits_{j=1}^q\overline{\psi}_{pl_j}\psi_{ps_j},
\quad \overline{l}=(l_1,\ldots,l_q),\quad \overline{s}=(s_1,\ldots,s_q).
\end{equation}
Now we use
\begin{lemma}\label{l:dop_slag}
Set $A=\{a_{i,j}\}_{i,j=1}^{2k}$, $b=\{b_{\overline{l},\overline{s}}\}$, where $\overline{l},\overline{s}$ is defined in
(\ref{sig}). Let $\Phi_r(A,b)$ be an analytic function
of the variables $\{a_{i,j}\}$ and $\{b_{\overline{l},\overline{s}}\}$ and let $(1-\varepsilon)n<r\le n$, $0\le l<\varepsilon n$
with some sufficiently small $\varepsilon>0$. Then there exist an absolute
constants $C_0, C_1$ such that
\begin{equation*}
\intd \Phi_r(n^{-1}\Psi^{(r)}_{2k},n^{-1}\sigma^{(r)})\widetilde{\mu}_{2k,l}^{(r)}(\Psi)d\,\Psi_{2k,r}\le
C_0\max\limits_{|a_{i,j}|, |b_{\overline{l},\overline{s}}|\le C_1}|\Phi_r(A,b)|\cdot \intd
\widetilde{\mu}_{2k,l}^{(r)}(\Psi)d\,\Psi_{2k,r},
\end{equation*}
where
\begin{equation}\label{mu_til_2}
\widetilde{\mu}_{2k,l}^{(r)}(\Psi)=e^{\Tr\Psi_{2k}^{(r)}\Lambda_{2k}}\,\mdet^{m-l} Q_{2k}^{(r)}.
\end{equation}
\end{lemma}
The proof of Lemma \ref{l:dop_slag} is given in Subsection 2.4.

Denote the expression multiplied by $\kappa_4$ in the r.h.s. of (\ref{usr_a_k}) by $n^{-1}X$. Write
\begin{multline}\label{bin}
\Big(\mdet\,Q^{(n)}_{2k}
+\frac{\kappa_4}{n}X+n^{-2}\Phi(\Psi)\Big)^m\\
=\sum\limits_{k_1+k_2\le m}\dfrac{m!}{k_1!k_2!(m-k_1-k_2)!}\Big(\mdet\,Q^{(n)}_{2k}\Big)^{m-k_1-k_2}
\Big(\frac{\kappa_4}{n}X\Big)^{k_1}\Big(n^{-2}\Phi(\Psi)\Big)^{k_2}.
\end{multline}
It is easy to see  that the terms in (\ref{bin}) such that $k_1+k_2\ge \varepsilon m$ give the contribution
of order $e^{-\varepsilon n\log n}$ and thus can be omitted. Hence, (\ref{ne_usr_k}), (\ref{usr_a_k}), and Lemma \ref{l:dop_slag} yield
\begin{equation}\label{usr_k_1}
\begin{array}{c}
F_{2k}(\Lambda_{2k})
=(1+O(n^{-1}))\displaystyle\int d\,\Psi_{2k,n}\,\,
 e^{\Tr\Psi_{2k}^{(n)}\Lambda_{2k}}\\
\bigg(\mdet\,Q^{(n)}_{2k}
+\dfrac{2\kappa_4}{n^2}\sum\limits_{l_1<l_2,s_1<s_2}\mdet (Q^{(n)}_{2k})^{(l_1,l_2;s_1,s_2)}
\sum\limits_{p=1}^n\overline{\psi}_{pl_1}
\psi_{ps_1}\overline{\psi}_{pl_2}\psi_{ps_2}\bigg)^m,
\end{array}
\end{equation}
where $Q^{(n)}_{2k}$ and $\Psi_{2k}^{(n)}$ are
defined in (\ref{Psi_2}).

  To compute the r.h.s. of (\ref{usr_k_1}) we use the Newton binomial formula and observe that the term with $p_s=p_l$ in
the product $\prod_{j=1}^q
\Big(n^{-1}\sum_{p_j}\overline{\psi}_{p_jl_{1,j}}\psi_{p_js_{1,j}}\overline{\psi}_{p_jl_{2,j}}\psi_{p_js_{2,j}}\Big)$
can be expressed in terms of (\ref{sig}) with an additional factor $n^{-1}$. Therefore, according to Lemma
\ref{l:dop_slag} it suffices to consider only the terms with $p_s\ne p_l$ or, taking into account the
symmetry, the term $p_1=n$, $p_2=n-1$, \ldots, $p_q=n-q+1$ with coefficient $n!/(n-q)!$.
Thus, we can write
\begin{equation}\label{F_int0}
\begin{array}{c}
F_{2k}(\Lambda_{2k})=(1+O(n^{-1}))
\displaystyle\sum\limits_{q=0}^m\binom{m}{q}\dfrac{n!}{(n-q)!}\dfrac{(2\kappa_4)^q}{n^q}I_{2k,q},
%
%
\end{array}
\end{equation}
where
\begin{equation}\label{I_2k,q}
I_{2k,q}=\displaystyle\int
\widetilde{\mu}_{2k,q}^{(n)}(\Psi)P_{q,n}^{(n)}(\Psi)d\,\Psi_{2k,n},
\end{equation}
\begin{equation}\label{P_q}
P_{q,l}^{(r)}(\Psi)=\prod\limits_{p=n-q+1}^l
\bigg(\sum\limits_{l_1^p<l_2^p,s_1^p<s_2^p}\mdet (Q^{(r)}_{2k})^{(l_1^p,l_2^p;s_1^p,s_2^p)}
\overline{\psi}_{pl_1^p}
\psi_{ps_1^p}\overline{\psi}_{pl_2^p}\psi_{ps_2^p}\bigg),
\end{equation}
where $\widetilde{\mu}_{2k,q}^{(n)}$ is defined in (\ref{mu_til_2}).
Note that
\begin{equation*}
\mdet (Q^{(n)}_{2k})^{(l_1,l_2;s_1,s_2)}=\mdet
(Q^{(n-q)}_{2k})^{(l_1,l_2;s_1,s_2)}+n^{-1}\widetilde{\Phi}_1(\Psi),
\end{equation*}
$\widetilde{\Phi}_1(\Psi)$ are polynomials of $\{\psi_{js}\}_{j,s=1}^{n\,2k}$, $\{\overline{\psi}_{js}\}_{j,s=1}^{n\,2k}$
with the sum of the coefficients of order $O(q)$, $m,n\to\infty$.
Hence, using Lemma \ref{l:dop_slag}, we get
\begin{equation}\label{I_2k,q_1}
I_{2k,q}=\displaystyle\int
\widetilde{\mu}_{2k,q}^{(n)}(\Psi)P_{q,n}^{(n-q)}(\Psi)d\,\Psi_{2k,n}(1+O(q/n))=:\widetilde{I}_{2k,q}(1+O(q/n)),
\end{equation}
According to Lemma \ref{l:repr_det} and (\ref{G_Gr_1}) we can rewrite $\widetilde{I}_{2k,q}$ as
\begin{align}\notag
\widetilde{I}_{2k,q}&=K_{2k,m-q}\int d\,\mu(V)
d\,\Psi_{2k,n-q}\dfrac{e^{\hbox{tr}\,\Lambda_{2k}\Psi_{2k}^{(n-q)}+\hbox{tr}\,Q_{2k}^{(n-q)}V}}
{\det^{m-q+2k}V}\\ \label{I_2k,q_2}
&\times \int \prod\limits_{p=n-q+1}^n\prod\limits_{l=1}^{2k}d\,\overline{\psi}_{pl}d\,\psi_{pl}
e^{\sum\limits_{i,j=1}^{2k}(\Lambda_{2k}-n^{-1}V)_{i,j}\sum\limits_{p=n-q+1}^n\overline{\psi}_{pi}\psi_{pj}}
P_{q,n}^{(n-q)}(\Psi)\\ \notag
&=K_{2k,m-q}\int d\,\mu(V)
d\,\Psi_{2k,n-q}\dfrac{e^{\hbox{tr}\,\Lambda_{2k}\Psi_{2k}^{(n-q)}+\hbox{tr}\,Q_{2k}^{(n-q)}V}}
{\det^{m-q+2k}V}\\ \notag
&\times\bigg(\sum\limits_{l_1<l_2,s_1<s_2}\mdet (Q^{(n-q)}_{2k})^{(l_1,l_2;s_1,s_2)}\mdet (\Lambda_{2k}-n^{-1}V)_{(l_1,l_2;s_1,s_2)}
\bigg)^{q}.
\end{align}
Besides, the Cauchy-Binet formula (see \cite{Gant:59}) yields for $2k\times 2k$ matrices $A,B$
\[
\sum\limits_{l_1<l_2,s_1<s_2}\mdet A^{(l_1,l_2;s_1,s_2)}B_{(l_1,l_2;s_1,s_2)}
=\dfrac{1}{2}\dfrac{d^2}{dx^2}\det(x-AB)\Big|_{x=0}.
\]
Thus, using again Lemma \ref{l:repr_det} and Remark 1.3, we obtain
\begin{align}\notag
I_{2k,q}&=
K_{2k,m-q}\int d\,\mu(V)
d\,\Psi_{2k,n-q}\dfrac{e^{\hbox{tr}\,\Lambda_{2k}\Psi_{2k}^{(n-q)}+\hbox{tr}\,Q_{2k}^{(n-q)}V}}
{\det^{m-q+2k}V}\\ \label{I_2k,q_3}
&\int\prod\limits_{s=1}^q\dfrac{d\,z_s}{2\pi i z_s^3}\int\prod\limits_{s=1}^q\dfrac{K_{2k,1}d\,\mu(W_s)}
{\det^{2k+1}W_s}e^{\hbox{tr}\,Q^{(n-q)}_{2k}(\Lambda_{2k}-n^{-1}V)\sum\limits_{p=1}^qW_p-
\sum\limits_{p=1}^qz_p\hbox{tr}\,W_p}\\ \notag
&=K_{2k,m-q}K_{2k,n-q}\int d\,\mu(V)
d\,\mu(U)\dfrac{e^{\hbox{tr}\,\Lambda_{2k}U+\hbox{tr}\,(1-n^{-1}U)V}}
{\det^{m-q+2k}V\det^{n-q+2k}U}\\ \notag
&\int\prod\limits_{s=1}^q\dfrac{d\,z_s}{2\pi i z_s^3}\int\prod\limits_{s=1}^q\dfrac{K_{2k,1}d\,\mu(W_s)}
{\det^{2k+1}W_s}e^{\hbox{tr}\,(1-n^{-1}U)(\Lambda_{2k}-n^{-1}V)\sum\limits_{p=1}^qW_p-
\sum\limits_{p=1}^qz_p\hbox{tr}\,W_p}\\ \notag
&=K_{2k,n-q}\int \dfrac{e^{\hbox{tr}\,\Lambda_{2k}U}\det^{m-q}(1-n^{-1}U)
\det^{m-q}(1-n^{-1}\sum\limits_{p=1}^qW_p)}
{\det^{n-q+2k}U}\\ \notag
&\times e^{\hbox{tr}\,(1-n^{-1}U)\Lambda_{2k}\sum\limits_{p=1}^qW_p-
\sum\limits_{p=1}^qz_p\hbox{tr}\,W_p}\prod\limits_{s=1}^q\dfrac{d\,z_s}{2\pi i z_s^3}\prod\limits_{s=1}^q
\dfrac{K_{2k,1}d\,\mu(W_s)}
{\det^{2k+1}W_s}d\,\mu(U).
\end{align}
Besides,
\[
\mdet^{m-q}(1-n^{-1}\sum\limits_{p=1}^qW_p)=e^{-c_{m,n}\sum\limits_{p=1}^q\hbox{tr}\,W_p}(1+O(q/n)).
\]
Substituting this to (\ref{I_2k,q_3}) and using (\ref{repr_det}), we get
\begin{align}\notag
\widetilde{I}_{2k,q}&=(1+O(q/n))
K_{2k,n-q}\int \dfrac{e^{\hbox{tr}\,\Lambda_{2k}U}\det^{m-q}(1-n^{-1}U)}
{\det^{n-q+2k}U}\\ \label{I_2k,q_4}
&\times e^{\hbox{tr}\,(1-n^{-1}U)\Lambda_{2k}\sum\limits_{p=1}^qW_p-
\sum\limits_{p=1}^q(z_p+c_{m,n})\hbox{tr}\,W_p}\prod\limits_{s=1}^q\dfrac{d\,z_s}{2\pi i z_s^3}\prod\limits_{s=1}^q
\dfrac{K_{2k,1}d\,\mu(W_s)}
{\det^{2k+1}W_s}d\,\mu(U)\\ \notag
&=(1+O(q/n))
K_{2k,n-q}\int \dfrac{e^{\hbox{tr}\,\Lambda_{2k}U}\det^{m-q}(1-n^{-1}U)}
{\det^{n-q+2k}U}S_2^q((1-n^{-1}U)\Lambda_{2k})d\,\mu(U)
\end{align}
Recall that we are interested in $\Lambda_{2k}=\Lambda_{0,2k}+\widehat{\xi}/(na)^\beta$, where
$\Lambda_0=\hbox{diag}\{\lambda_0,..,\lambda_0\}$,
$\widehat{\xi}=\hbox{diag}\{\xi_1,..,\xi_{2k}\}$, and $a$, $\beta$ are defined in (\ref{a}), (\ref{bet}). Thus,
\begin{multline}\label{I_2k,q_5}
\widetilde{I}_{2k,q}=(1+O(n^{-\beta})+O(q/n))
K_{2k,n-q}\int \dfrac{e^{\hbox{tr}\,\Lambda_{2k}U}\det^{m-q}(1-n^{-1}U)}
{\det^{n-q+2k}U}\\
\times S_2^q((1-n^{-1}U)\Lambda_{0})d\,\mu(U),
\end{multline}
Let us change variables to $U=W^{*}V W$, where $W$ is a unitary $2k\times 2k$ matrix and
$V=\hbox{diag}\{v_1,\ldots,v_{2k}\}$. Since functions $\mdet (I-n^{-1}U)$,
$S_2((I-n^{-1}U)\Lambda_0)$,
and $\mdet\,U$
are unitary invariant, (\ref{I_2k,q_5}) implies
\begin{equation}\label{I_2k,q_6}
\begin{array}{c}
I_{2k,q}=(1+O(n^{-\beta})+O(q/n))
K_{2k,n-q}\displaystyle\oint\limits_{\omega}\prod\limits_{j=1}^{2k}\dfrac{d\,v_j}{2\pi i}\int
d\,\mu(W)e^{\Tr W^{*}V W\Lambda_{2k}}\Delta^2(V)\\
\dfrac{\mdet^{m-q} (I-n^{-1}V)}{\mdet^{n-q+2k} V }
S_2^q((I-n^{-1}V)\Lambda_0).
\end{array}
\end{equation}
where $\omega$ is any closed contour encircling $0$.
The integral over the unitary group $U(2k)$ can be computed using the
Harish Chandra/Itsykson-Zuber formula (\ref{Its-Zub}). Shifting $v_i\to nv_i$, we obtain
\begin{equation}\label{I_2k,q_fin}
\begin{array}{c}
\widetilde{I}_{2k,q}=\dfrac{S_{2k}K_{2k,n-q}(n^{\beta-1}a^{\beta})^{k(2k-1)}}{n^{2k(n-q)}D_{2k}}
\displaystyle\oint\limits_{\omega}\prod\limits_{j=1}^{2k}\dfrac{d\,v_j}{2\pi i}
e^{n\Tr V \Lambda_{0}+n^{1-\beta}a^{-\beta}\Tr V \widehat{\xi}}\dfrac{\Delta(V)}{\Delta(\widehat{\xi})}\\
\prod\limits_{l=1}^{2k}\dfrac{(1-v_l)^{m-q}}{v_l^{n-q+2k}}S_2^q(I-V)(1+O(n^{-\beta})+O(q/n)).
\end{array}
\end{equation}
Hence, since
\[
\dfrac{n!}{(n-q)!}\cdot \dfrac{\prod_{s=0}^{2k-1}(n-q+s)!}{n^{2k(n-q)+q}}=(2\pi)^k n^{2k^2} e^{-2kn}
(1+O(1/n)),
\]
we get (\ref{F_int}) from (\ref{F_int0}), (\ref{I_2k,q_1}), and (\ref{I_2k,q_fin}).
\subsection{Proofs of Lemmas \ref{l:repr_det}, \ref{l:dop_slag}}
\textbf{Proof of Lemma \ref{l:repr_det}} Let $A$ be a normal matrix. Then we can set $A=V_0^*A_0V_0$ and $U=W^*U_0W$, where $A_0=\hbox{diag}\,(a_1,..,a_p)$,
$U_0=\hbox{diag}\,(u_1,..,u_p)$ and $V_0$, $W$ are the matrices diagonalizing $A$ and $U$ correspondingly.
We obtain
\[
I:=\displaystyle\int\dfrac{e^{\hbox{tr}\,AU}}{\hbox{det}^{p+l}U}
d\,U=\displaystyle\int\dfrac{e^{\hbox{tr}\,V_0^*A_0V_0W^*U_0W}\triangle^2(u_1,\ldots,u_p)}
{\prod\limits_{j=1}^pu_j^{p+l}}
d\,\mu(W)\prod\limits_{j=1}^p\dfrac{du_j}{2\pi i}.
\]
Shifting integration with respect to $W$ as $WV_0^*\to W$ and using (\ref{Its-Zub}), we obtain
\[
\begin{array}{c}
I=q_p\displaystyle\oint\limits_\omega\dfrac{e^{\sum\limits_{j=1}^pa_ju_j}
\triangle(u_1,\ldots,u_p)}{\triangle(A_0)\prod\limits_{j=1}^pu_j^{p+l}}\prod\limits_{j=1}^p
\dfrac{d\,u_j}{2\pi i}\\
=\dfrac{q_p}{\triangle(A_0)}\det\left[\oint\limits_\omega\dfrac{e^{a_ju_j}
}{u_j^{p+l-s}}\dfrac{d\,u_j}{2\pi i}\right]_{j,s=1,0}^{p,p-1}=\dfrac{q_p}{\triangle(A_0)}
\det\left[\dfrac{a_j^{p+l-s-1}}{(p+l-s-1)!}\right]_{j,s=1,0}^{p,p-1}\\
=\dfrac{q_p\triangle(1/a_1,\ldots,1/a_p)\prod\limits_{j=1}^pa_j^{p+l-1}}{\prod\limits_{s=0}^{p-1}
(p+l-s-1)!\triangle(A_0)}=\dfrac{(-1)^{\frac{p(p-1)}{2}}q_p\prod\limits_{j=1}^pa_j^l}{\prod\limits_{s=0}^{p-1}
(p+l-s-1)!}=\dfrac{(-1)^{\frac{p(p-1)}{2}}q_p\hbox{det}^l A}{\prod\limits_{s=0}^{p-1}
(p+l-s-1)!},
\end{array}
\]
and (\ref{repr_det}) is proved for the normal $A$.

Let now $A$ be an arbitrary matrix. According to the polar decomposition, we can write $A=SW$, where
$W$ is a unitary $p\times p$ matrix and $S$ is a diagonal $p\times p$ matrix. Since we proved (\ref{repr_det})
for any normal $A$, we proved it for $S=\hbox{diag}\,\{e^{i\alpha_1},\ldots, e^{i\alpha_p}\}$, $\alpha_1,..,
\alpha_p\in \mathbb{R}$. Besides, it is easy to see that both sides of (\ref{repr_det}) is  analytic
functions of the elements of $S$. Therefore, we are proved (\ref{repr_det}) for any $A$.

$\Box$
\medskip

\textbf{Proof of Lemma \ref{l:dop_slag}}
According to Lemma \ref{l:repr_det} and (\ref{mu_til_2}), we have
\begin{equation}\label{J}
\begin{array}{c}
\intd \widetilde{\mu}_{2k,l}^{(r)}(\Psi)d\,\Psi_{2k,r}=K_{2k,m-l}\int\intd\dfrac{e^{\hbox{tr}\Lambda_{2k}\Psi^{(r)}_{2k}+
\hbox{tr}(1-n^{-1}\Psi^{(r)}_{2k})V}}{\mdet^{m-l+2k}V}d\,\mu(V)d\,\Psi_{2k,r}\\
=K_{2k,m-l}\intd\dfrac{e^{\hbox{tr}V}\mdet^r(\Lambda_{2k}-n^{-1}V)}{\mdet^{m-l+2k}V}d\,\mu(V)=:J.
\end{array}
\end{equation}
It is proved below (see Section 3 and 4 (note that if we change $v\to \lambda_0(1-v)$ in $J$
we obtain the integral like in (\ref{F_int_V}) and (\ref{F_int_V_ed}))) that
\begin{equation}\label{ots_mer}
\big|J\big|\\
\ge
\frac{CK_{2k,m-l}}{n^{2k(m-l)}}\oint\limits_{\widetilde{\omega}} e^{n\sum\limits_{j=1}^{2k}\Re v_j+
n^{1-\beta}a^{-\beta}\sum\limits_{j=1}^{2k}\xi_j\Re v_j}\prod\limits_{j=1}^{2k}
|\lambda_0-v_j|^{r}|\Delta(V)|^{2}\prod\limits_{j=1}^{2k}\frac{|dv_j|}{|v_j|^{m-l+2k}},
\end{equation}
where $a$ and $\beta$ are defined in (\ref{a}) and (\ref{bet}) and
\begin{equation}\label{om_til}
\widetilde{\omega}=\Big\{z\in \mathbb{C}:|z|=\Big(\frac{m-l}{n}\lambda_0\Big)^{1/2}\Big\}.
\end{equation}
 Moreover
the integral outside of the any $n$-independent neighborhood of $v_{\pm}=(\lambda_0+\frac{m-l}{n}-1)/2\pm
\pi\lambda_0\rho(\lambda_0)$ give contribution $O(e^{-nC})$, hence we can can deform $\widetilde{\omega}$
near $z=\big(\frac{m-l}{n}\lambda_0\big)^{1/2}$ such that $|v-\lambda_0|>\delta$ on $\widetilde{\omega}$.
Thus, if we define
\begin{equation}\label{exp}
\langle (\ldots)\rangle=J^{-1}\intd (\ldots)
\widetilde{\mu}_{2k,l}^{(r)}(\Psi)d\,\Psi_{2k,r}
\end{equation}
the definition is correct.

Using (\ref{cont_int_gr}), we get
\begin{equation}\label{int_cont}
\begin{array}{c}
\langle\Phi_r\rangle:=\langle \Phi_r(n^{-1}\Psi^{(r)}_{2k},\sigma^{(r)})\rangle
=\displaystyle\oint\limits_{\Omega} \Phi_r(A,b)\prod\limits_{i,j=1}^{2k}\dfrac{d\,a_{i,j}}{2\pi i}
\prod\limits_{\overline{l},\overline{s}}\dfrac{d\,b_{\overline{l},\overline{s}}}{2\pi i}\\
\times\Big\langle \prod\limits_{i,j=1}^{2k}\dfrac{1}{a_{i,j}-n^{-1}(\Psi^{(r)}_{2k})_{i,j}}
\prod\limits_{\overline{l},\overline{s}}
\dfrac{1}{b_{\overline{l},\overline{s}}-n^{-1}\sigma^{(r)}_{\overline{l},\overline{s}}}
\Big\rangle.
\end{array}
\end{equation}
Thus, to prove Lemma \ref{l:dop_slag}, we have to estimate the expectation above.
Expanding the functions into the series with respect to $\{\Psi^{(r)}_{2k})_{i,j}\}$,
$\{\sigma^{(r)}_{\overline{l},\overline{s}}\}$, we get
\begin{equation}\label{int_1}
\begin{array}{c}
\Big\langle \prod\limits_{i,j=1}^{2k}\dfrac{1}{a_{i,j}-n^{-1}(\Psi^{(r)}_{2k})_{i,j}}
\prod\limits_{\overline{l},\overline{s}}\dfrac{1}{b_{\overline{l},\overline{s}}-n^{-1}\sigma^{(r)}_{\overline{l},\overline{s}}}
\Big\rangle\\
=\sumd\limits_{i,j,\overline{l},\overline{s}}\sumd\limits_{l_{i,j}=1}^{r}
\sum\limits_{t_{\overline{l},\overline{s}}=1}^{r}\Big\langle
\prod\limits_{i,j=1}^{2k}(n^{-1}(\Psi^{(r)}_{2k})_{i,j})^{l_{i,j}}\prod\limits_{\overline{l},\overline{s}}
(n^{-1}\sigma^{(r)}_{\overline{l},\overline{s}})^{t_{\overline{l},\overline{s}}}\Big\rangle
\prod\limits_{i,j=1}^{2k}a_{i,j}^{-l_{i,j}-1}\prod\limits_{\overline{l},\overline{s}}
b_{\overline{l},\overline{s}}^{-t_{\overline{l},\overline{s}}-1}\\
:=\sumd\limits_{i,j,\overline{l},\overline{s}}\sumd\limits_{l_{i,j}=1}^{r}
\sum\limits_{t_{\overline{l},\overline{s}}=1}^{r}M(\{l_{i,j}\}, \{t_{\overline{l},\overline{s}}\})\prod\limits_{i,j=1}^{2k}a_{i,j}^{-l_{i,j}-1}\prod\limits_{\overline{l},\overline{s}}
b_{\overline{l},\overline{s}}^{-t_{\overline{l},\overline{s}}-1}.
\end{array}
\end{equation}
To estimate the moments $\{M(\{l_{i,j}\}, \{t_{\overline{l},\overline{s}}\})\}$, we introduce the generating
function
\begin{equation}
F(\zeta,z):=\Big\langle \exp\Big\{n^{-1}\hbox{tr}\,\zeta\Psi^{(r)}_{2k}+n^{-1}\sum\limits_{\overline{l},\overline{s}}z_{\overline{l},\overline{s}}
\sigma^{(r)}_{\overline{l},\overline{s}}\Big\} \Big\rangle,
\end{equation}
where $\zeta=\{\zeta_{i,j}\}_{i,j=1}^{2k}$. It is easy to see that the derivatives $F(\zeta,z)$ with respect to
$\{\zeta_{i,j}\}$ and $\{z_{\overline{l},\overline{s}}\}$ will give us
the moments $\{M(\{l_{i,j}\}, \{t_{\overline{l},\overline{s}}\})\}$.

 Using Lemma \ref{l:repr_det} and then integrating over
$d\,\Psi_{2k,r}$, we obtain
%
\begin{equation*}
\begin{array}{c}
F(\zeta,z)=\dfrac{K_{2k,m-l}}{J}\intd\dfrac{e^{\hbox{tr}\,V+
\hbox{tr}\,(\Lambda_{2k}+n^{-1}\zeta-n^{-1}V)\Psi^{(r)}_{2k}+
n^{-1}\sum\limits_{\overline{l},\overline{s}}z_{\overline{l},\overline{s}}
\sigma^{(r)}_{\overline{l},\overline{s}}}}{\mdet^{m-l+2k}V}d\,\mu(V)d\,\Psi_{2k,r}\\
=\dfrac{K_{2k,m-l}}{J}\intd\dfrac{e^{\hbox{tr}\,V}}{\mdet^{m-l+2k}V}
\Phi_{1}^{r}(\Lambda_{2k}-n^{-1}V,n^{-1}\zeta,n^{-1}z)d\,\mu(V),
\end{array}
\end{equation*}
where $J$ is defined in (\ref{J}). Moreover, according to (\ref{G_Gr}) -- (\ref{G_Gr_1}), $\Phi_{1}(\Lambda_{2k}-n^{-1}V,n^{-1}\zeta,n^{-1}z)$ is a polynomial of the
entries of $\Lambda_{2k}-n^{-1}V$ and of $\{\zeta_{i,j}\}$, $\{z_{\overline{l},\overline{s}}\}$ with $n$-independent
coefficients and degree at most $2k$ and such that the degree of each variable in $\Phi_{1}(V,n^{-1}\zeta,n^{-1}z)$
is at most one. Here we also used that the integral over $d\,\Psi_{2k,r}$ can be factorized in
$\{\overline{\psi}_{pi}\psi_{pj}\}_{i,j=1}^{2k}$. Besides,
\begin{equation}\label{Phi}
\Phi_{1}(n^{-1}V,n^{-1}\zeta,n^{-1}z)=\mdet(\Lambda_{2k}-n^{-1}V)+\widetilde{f}(\Lambda_{2k}-n^{-1}V,n^{-1}\zeta,n^{-1}z),
\end{equation}
where $\widetilde{f}(\Lambda_{2k}-n^{-1}V,n^{-1}\zeta,n^{-1}z)$ contains all terms of $\Phi_1$ which includes
$\{\zeta_{i,j}\}$ or $\{z_{\overline{l},\overline{s}}\}$.

Recall that we are interested in $\Lambda=\Lambda_0+\widehat{\xi}/(na)^\beta$, where
$\Lambda_0=\hbox{diag}\{\lambda_0,..,\lambda_0\}$,
$\widehat{\xi}=\hbox{diag}\{\xi_1,\ldots,\xi_{2m}\}$, and $a$ and $\beta$ are defined in (\ref{a}) and (\ref{bet}).
Change the variables $v_j\to nv_j$, $j=1,..,2k$, where $\{v_j\}$ are the eigenvalues of $V$,
and replace the integration over the unit circle by the integration over $\widetilde{\omega}$ of (\ref{om_til})
\begin{equation}\label{F_pr}
F(\zeta,z)=\frac{K_{2k,m-l}}{J\cdot n^{2k(m-l)}}\intd\dfrac{e^{n\hbox{tr}\,V}}{\mdet^{m-l+2k}V}
\Phi_{1}^{r}(\Lambda_{2k}-V,n^{-1}\zeta,n^{-1}z)d\,\mu(V)
\end{equation}
We have from the description of $\Phi_1$ and (\ref{Phi})
\begin{equation}\label{ogr_Phi2}
|\Phi_1(\Lambda_{2k}-V,n^{-1}\zeta,n^{-1}z)|\le C|\mdet(\Lambda_{0}-V)|\displaystyle\prod\limits_{i,j=1}^{2k}
\Big(1+\dfrac{C(V)|\zeta_{i,j}|}{n}\Big)
\displaystyle\prod\limits_{\overline{l},\overline{s}}\Big(1+\dfrac{C(V)|z_{\overline{l},\overline{s}}|}{n}\Big),
\end{equation}
where $C(V)>0$ is bounded for $v_j\in \widetilde{\omega}$ with $\widetilde{\omega}$ of (\ref{om_til}) (recall
that we can deform $\widetilde{\omega}$ near $z=\big(\frac{m-l}{n}\lambda_0\big)^{1/2}$ such
 that $|v-\lambda_0|>\delta$ on $\widetilde{\omega}$ ). Since $\{M(\{l_{i,j}\}, \{t_{\overline{l},\overline{s}}\})\}$
 are the derivatives of the generating function, we can write
\begin{equation}
M(\{l_{i,j}\}, \{t_{\overline{l},\overline{s}}\})=\Big\langle \prod\limits_{i,j=1}^{2k}\oint\limits_{\Omega_{i,j}}
\dfrac{l_{i,j}!}{2\pi i} \dfrac{d\,\zeta_{i,j}}{\zeta_{i,j}^{l_{i,j}+1}}
\prod\limits_{\overline{l},\overline{s}}\oint\limits_{\Sigma_{\overline{l},\overline{s}}}
\dfrac{t_{\overline{l},\overline{s}}!}{2\pi i}\dfrac{d\,z_{\overline{l},\overline{s}}}{z_{\overline{l},\overline{s}}^{t_{\overline{l},\overline{s}}+1}}
 F(\zeta,z)
\Big\rangle.
\end{equation}
This, (\ref{F_pr}), and (\ref{ogr_Phi2}) yield
\begin{equation}\label{ineq_1}
|M(\{l_{i,j}\}, \{t_{\overline{l},\overline{s}}\})|\le \prodd\limits_{i,j=1}^{2k} \min\limits_{t
\in\Omega_{i,j}}l_{i,j}!e^{C|t|-l_{i,j}\log |t|}
\prodd\limits_{\overline{l},\overline{s}}
\min\limits_{t\in \Sigma_{\overline{l},\overline{s}}}
t_{\overline{l},\overline{s}}!e^{C|t|-t_{\overline{l},\overline{s}}
\log |t|}
\end{equation}
Choose $\Omega_{i,j}=\{\zeta\in \mathbb{C}: |\zeta|=l_{i,j}/C\}$,
$\Sigma_{\overline{l},\overline{s}}=\{z\in\mathbb{C}:|z|=t_{\overline{l},\overline{s}}/C\}$.
Then (\ref{ineq_1}) yields
\begin{equation}\label{ineq_2}
\begin{array}{c}
|M(\{l_{i,j}\}, \{t_{\overline{l},\overline{s}}\})|\le \prodd\limits_{i,j=1}^{2k} \sqrt{2\pi l_{i,j}}C^{l_{i,j}}
\prodd\limits_{\overline{l},\overline{s}}\sqrt{2\pi t_{\overline{l},\overline{s}}}
C^{t_{\overline{l},\overline{s}}}\\
=\prodd\limits_{i,j=1}^{2k} \sqrt{2\pi l_{i,j}}C^{l_{i,j}}
\prodd\limits_{\overline{l},\overline{s}}\sqrt{2\pi t_{\overline{l},\overline{s}}}
C^{t_{\overline{l},\overline{s}}}.
\end{array}
\end{equation}
Thus, if $|a_{i,j}|>C$ and $|b_{\overline{l},\overline{s}}|>C$ in (\ref{int_1}) we obtain
Lemma \ref{l:dop_slag} from (\ref{J}), (\ref{int_cont}) -- (\ref{int_1}).
$\Box$
\medskip

\section{Asymptotic analysis in the bulk of the spectrum.}
In this section we prove Theorem \ref{thm:1}, passing to the limit $m,n\to\infty$ in (\ref{I_2k,q_fin}) for
$\lambda_j=\lambda_0+\xi_j/n\rho(\lambda_0)$, where $\rho$ is defined in
(\ref{rho_mp}), $\lambda_0\in \sigma$ with $\sigma$ of (\ref{lam_pm}), and $\xi_j\in [-M,M]\subset\mathbb{R}$, $j=1,..,2k$.

To this end consider the function
\begin{equation}\label{V}
V(v,\lambda_0)=-\lambda_0v-c_{m,n}\log (1-v)+\log v +S^*,
\end{equation}
where
\begin{equation}\label{c,S}
c_{m,n}=\dfrac{m}{n},\quad
S^*=\dfrac{\lambda_0-c_{m,n}+1}{2}+\frac{c_{m,n}}{2}
\log\frac{c_{m,n}}{\lambda_0}-
\frac{1}{2}\log\frac{1}{\lambda_0}.
\end{equation}

Then (\ref{F_int}) and (\ref{I_2k,q_fin}) yield
\begin{equation}\label{F_int_V}
D_{2k}^{-1}(\lambda_0)n^{-k^2}F_{2k}(\Lambda_{2k})=Z_{2k}
\oint\limits_{\omega_0} W_{n}(v_1,\ldots,v_{2k})\prod\limits_{j=1}^{2k}d\,v_j(1+o(1)),
\end{equation}
where $D_{2k}$ is defined in (\ref{D_2k}),
\begin{equation}\label{W_n}
\begin{array}{c}
W_{n}(v_1,\ldots,v_{2k})=
e^{-n\sum\limits_{l=1}^{2k}V(v_l,\lambda_0)+\sum\limits_{l=1}^{2k}\frac{\xi_l}{\rho(\lambda_0)} v_l}\dfrac{\triangle(V)}
{\triangle(\widehat{\xi})}\prod\limits_{j=1}^{2k}\dfrac{1}{v_j^{2k}}\\
\times
\exp\Big\{2c_{m,n}\kappa_4S_2((I-V)\Lambda_0)
\prod\limits_{l=1}^{2k}\frac{v_l}{1-v_l}\Big\},
\end{array}
\end{equation}
and
\begin{equation}\label{Z_kp}
Z_{2k}=\dfrac{n^{k^2}\rho(\lambda_0)^{k(k-1)} e^{-2k\kappa_4-\alpha(\lambda_0)\sum\limits_{j=1}^{2k}\xi_j}
}{2^{2k}\pi^{2k}c^{k/2}}.
\end{equation}
Now we need the following lemma
\begin{lemma}\label{l:min_L}
The function $\Re V(v,\lambda_0)$ for $v=\lambda_0^{-1/2}e^{i\varphi},$ $\varphi\in (-\pi,\pi]$ attains its minimum at
\begin{equation}\label{v_pm}
v=v_{\pm}:=\lambda_0^{-1/2}e^{\pm i\varphi_0}:=\dfrac{\lambda_0-c_{m,n}+1}{2\lambda_0}\pm i\pi\rho(\lambda_0).
\end{equation}
Moreover, if $\varphi\not\in U_n(\pm\varphi_0):=(\pm\varphi_0-n^{-1/2}\log n,
\pm\varphi_0+n^{-1/2}\log n)$, then we have for sufficiently big $n$
\begin{equation}\label{ineqv_ReV}
\Re V(\lambda_0^{-1/2}e^{i\varphi},\lambda_0)\ge \dfrac{C\log^2n}{n}.
\end{equation}
\end{lemma}
\begin{proof}
Note that for $\varphi\in (-\pi,\pi]$
\begin{equation}\label{ReV}
\Re V(\lambda_0^{-1/2}e^{i\varphi},\lambda_0)=-\lambda_0^{1/2}\cos\varphi-\dfrac{c_{m,n}}{2}\log\left(1+
\lambda_0^{-1}-2\lambda_0^{-1/2}\cos\varphi\right)+\log\lambda_0^{-1}+S^*,
\end{equation}
where $S^*$ and $c_{m,n}$ are defined in (\ref{c,S}). Thus
\begin{align}\notag
\dfrac{d}{d\,\varphi}\Re V(\lambda_0^{-1/2}e^{i\varphi},\lambda_0)&=
\lambda_0^{1/2}\sin\varphi\Big(1-\dfrac{c_{m,n}/\lambda_0}{1+
\lambda_0^{-1}-2\lambda_0^{-1/2}\cos\varphi}\Big),\\ \label{ReV_pr}
\dfrac{d^2}{d\,\varphi^2}\Re V(\lambda_0^{-1/2}e^{i\varphi},\lambda_0)&=
\lambda_0^{1/2}\cos\varphi\Big(1-\dfrac{c_{m,n}/\lambda_0}{1+
\lambda_0^{-1}-2\lambda_0^{-1/2}\cos\varphi}\Big)\\ \notag
&+\dfrac{2c_{m,n}\sin^2\varphi/\lambda_0}{(1+
\lambda_0^{-1}-2\lambda_0^{-1/2}\cos\varphi)^2},
\end{align}
and $\varphi=\pm\varphi_0$ of (\ref{v_pm}) are the minimum points of
$\Re V(\lambda_0^{-1/2}e^{i\varphi},\lambda_0)$.
Writing
\begin{equation}\label{V_pm}
V_\pm:=V(v_{\pm},\lambda_0)=\mp i\lambda_0^{-1/2}\sin\varphi_0\pm i\varphi_0\pm
ic_{m,n}\arcsin \dfrac{\lambda_0^{-1/2}\sin\varphi_0}{1+
\lambda_0^{-1}-2\lambda_0^{-1/2}\cos\varphi_0},
\end{equation}
we conclude that
\[
\Re V(v_{\pm},\lambda_0)=0.
\]
Expanding $\Re V(\lambda_0^{-1/2}e^{i\varphi},\lambda_0)$ into the Taylor series and using (\ref{ReV_pr}) -- (\ref{V_pm}),
we obtain for $\varphi\in U_n(\pm\varphi_0)$:
\begin{equation}\label{ineqv_ne0}
\Re V(\lambda_0^{-1/2}e^{i\varphi},\lambda_0)=\dfrac{(\pi\lambda_0\rho(\lambda_0))^2}{c_{m,n}}(\varphi\mp \varphi_0)^2+O(n^{-3/2}\log^3n),
\end{equation}
where $\varphi_0$ is defined in (\ref{v_pm}). This and $c_{m,n}\to c,\,m,n\to\infty$ imply for $\varphi\not\in U_n(\pm\varphi_0)$
\[
\Re V(\lambda_0^{-1/2}e^{i\varphi},\lambda)\ge \dfrac{C\log^2n}{n}.
\]
The lemma is proved.
\end{proof}
Note that $|v_j|=\lambda_0^{-1/2}$, $j=1,..,2k$. Since $\xi_1,\ldots,\xi_{2k}$ are distinct,
the inequality $|\triangle(T)/\triangle(\widehat{\xi})|\le C_1$ and (\ref{ineqv_ReV}) yield
\begin{equation*}
\bigg|Z_{2k}\oint\limits_{\omega_0\setminus(U_{v,+}\cup U_{v,-})}\oint\limits_{\omega_0}..\oint\limits_{\omega_0}
W_n(v_1,\ldots,v_{2k})\prod\limits_{j=1}^{2k}d\,v_j
\bigg|\le C_1n^{k^2}e^{-C_2\log^2n},
\end{equation*}
where
\begin{equation}\label{om_0}
\omega_0=\{z\in \mathbb{C}: |z|=\lambda_0^{-1/2}\},
\end{equation}
$W_n$ and $Z_{2k}$ are defined in (\ref{W_n}) and
(\ref{Z_kp}) respectively, and
\begin{eqnarray}\label{U_pm}
U_{\pm}&=&\{\varphi\in (-\pi,\pi]: |\pm\varphi_0-\varphi|\le n^{-1/2}\log n\},\\ \notag
U_{v,\pm}&=&\{z=\lambda_0^{-1/2}e^{i\varphi}|\varphi\in U_{\pm}\}
\end{eqnarray}
with $\varphi_0$ of (\ref{v_pm}).

  Note that we have for $\varphi\in U_{\pm}$ in view of (\ref{V}) and (\ref{V_pm}) as $m,n\to\infty$
\begin{equation}\label{as_razl_V}
V(\lambda_0^{-1/2}e^{i\varphi},\lambda_0)=V_{\pm}+\left(\dfrac{1}{v_{\pm}^2}-\dfrac{c_{m,n}}{(1-v_\pm)^2}\right)\lambda_0^{-1}e^{\pm2i\varphi_0}
\dfrac{(\varphi \mp\varphi_0)^2}{2}+
f_\pm(\varphi\mp \varphi_{0}),
\end{equation}
where $f_\pm(\varphi\mp \varphi_{0})=O((\varphi\mp \varphi_{0})^3)$.
Shifting $\varphi_j\mp \varphi_0\to \varphi_j$ for $\varphi_j\in U_{\pm}$ and using (\ref{V_pm})
we obtain
\begin{equation}\label{int_okr1}
\begin{array}{c}
\dfrac{D_{2k}^{-1}(\lambda_0)}{(n\rho(\lambda_0))^{k^2}}
F_{2k}(\Lambda_{2k})=Z_{2k}\lambda_0^{k(2k-1)/2}\sum\limits_{s=1}^{2k}\sum\limits_{\alpha^s}
\intd\limits_{(U_{n})^{2k}}e^{G_s(\varphi_1,\ldots,\varphi_{2k})+\sum\limits_{j=1}^{2k}d_{\alpha_j}(\varphi_j)\xi_j}
\dfrac{\Delta(V^{\alpha^s})}{\Delta(\widehat{\xi})}\\
\prod\limits_{l=1}^{s}e^{-\frac{nc_{+}}{2}\varphi_l^2-(2k-1)i(\varphi_j+\varphi_0)}\prod\limits_{j=s+1}^{2k}
e^{-\frac{nc_{-}}{2}\varphi_j^2-(2k-1)i(\varphi_j-\varphi_0)}\prod\limits_{j=1}^{2k}d\,\varphi_j
=\sumd\limits_{s=1}^{2k}
\sumd\limits_{\alpha^s}T_{s,\alpha},
\end{array}
\end{equation}
where $\alpha^s=\{\alpha_j\}_{j=1}^{2k}$ is a permutation of $s$ pluses and $2k-s$ minuses,
\begin{align}\label{d,G}
&d_{\pm}(\varphi_j)=\frac{e^{i(\varphi_j\pm\varphi_0)}}
{\sqrt{\lambda_0}\rho(\lambda_0)},\quad V^{s}=\hbox{diag}\{e^{i(\varphi_1+\varphi_0)},..,e^{i(\varphi_s+\varphi_0)},
e^{i(\varphi_{s+1}-\varphi_0)},..,e^{i(\varphi_{2k}-
\varphi_0)}\},\\ \notag
&G_s(\varphi_1,\ldots,\varphi_{2k})=2c_{m,n}\kappa_4S_2((I-V^{s})\Lambda_0)
\prod\limits_{l=1}^{s}\frac{\lambda_0^{-1/2}e^{i(\varphi_l+\varphi_0)}}
{1-\lambda_0^{-1/2}e^{i(\varphi_l+\varphi_0)}}\prod\limits_{r=s+1}^{2k-s}\frac{\lambda_0^{-1/2}e^{i(\varphi_r-\varphi_0)}}
{1-\lambda_0^{-1/2}e^{i(\varphi_r-\varphi_0)}}\\ \notag
&-n\sum\limits_{l=1}^{s}(f_+(\varphi_l)+V_+(\varphi_l))-n\sum\limits_{r=s+1}^{2k-s}(f_-(\varphi_r)+V_-(\varphi_r)),
\end{align}
\begin{align}\label{cp}
c_{\pm}&=\left(\dfrac{1}{v_{\pm}^2}-\dfrac{c_{m,n}}{(1-v_\pm)^2}\right)\lambda_0^{-1}e^{\pm2i\varphi_0},
\quad U_{n}=(-n^{-1/2}\log n, n^{-1/2}\log n),\\ \notag
V^{\alpha^s}&=\hbox{diag}\{e^{i(\varphi_1+\alpha_1\varphi_0)},\ldots,e^{i(\varphi_{2k}+
\alpha_{2k}\varphi_0)}\}.
\end{align}
Define
\begin{align}\notag
I_s&:=\intd\limits_{\Omega_{n,s}}
e^{\frac{1}{\sqrt{n}}\sum\limits_{j=1}^s\xi_jg(\varphi_j)}
\prod\limits_{j<l}(\varphi_j-\varphi_l)d\,\nu_{s}(\varphi_1,\ldots,\varphi_s)\\ \label{I}
&=\intd\limits_{\Omega_{n,s}}\det\Big\{e^{\frac{1}{\sqrt{n}}\xi_jg(\varphi_j)}
\varphi_j^{l-1}\Big\}_{j,l=1}^sd\,\nu_{s}(\varphi_1,\ldots,\varphi_s)\\ \notag
&=\sum\limits_{p_1,..,p_s=0}^\infty\intd\limits_{\Omega_{n,s}}\det\Big\{
(n^{-1/2}\xi_jg(\varphi_j))^{p_j}\varphi_j^{l-1}/p_j!
\Big\}_{j,l=1}^sd\,\nu_{s}(\varphi_1,\ldots,\varphi_s),
\end{align}
where $d\,\nu_{s}(\varphi_1,\ldots,\varphi_s)$ is a measure on $\Omega_{n,s}:=(-\log n,\log n)^s$ which is symmetric
in\\
 $(\varphi_1,..,\varphi_s)$ and $g(\varphi)$ is a function such that $g(\varphi)=C\varphi(1+o(1))$,
$n\to\infty$. Note that if we take the term of (\ref{I}) such that $p_{s_1}=p_{s_2}$, $s_1\ne s_2$, then this term
is zero since $d\,\nu_{s}(\varphi_1,\ldots,\varphi_s)$ is symmetric in $(\varphi_1,\ldots,\varphi_s)$.
Moreover, the order of
$$\det\big\{(n^{-1/2}\xi_jg(\varphi_j))^{p_j}\varphi_j^{l-1}/p_j!
\big\}_{j,l=1}^s$$ is $n^{-(p_1+..+p_s)/2}$ and if $\{p_1,..,p_s\}\ne \{0,1,..,s-1\}$ the order is less than
$n^{-s(s-1)/2}$. Hence, denoting by $\widetilde{\sum}$
the sum over all permutations $\{p_1,..,p_s\}$ of $\{0,1,..,s-1\}$, we obtain
\begin{align}\notag
I_s&=\frac{n^{-s(s-1)/2}}{\prod_{j=0}^{s-1}j!}\widetilde{\sum}
\intd\limits_{\Omega_{n,s}}\prod\limits_{j=1}^s(\xi_jg(\varphi_j))^{p_j}
\triangle(\varphi_1,..,\varphi_s)d\,\nu_{s}(\varphi_1,\ldots,\varphi_s) (1+o(1))\\  \label{I1}
&=\frac{\triangle(\xi_1,..,\xi_s)}{n^{s(s-1)/2}\prod_{j=0}^{s-1}j!}
\intd\limits_{\Omega_{n,s}}
\prod\limits_{j=1}^{s}g(\varphi_j)^{j-1}\triangle(\varphi_1,..,\varphi_s)d\,\nu_{s}(\varphi_1,\ldots,\varphi_s)
 (1+o(1))\\ \notag
&=\frac{\triangle(\xi_1,..,\xi_s)}{n^{s(s-1)/2}\prod_{j=0}^{s}j!}
\intd\limits_{\Omega_{n,s}}
\triangle(g(\varphi_1),..,g(\varphi_s))\triangle(\varphi_1,..,\varphi_s)d\,\nu_{s}(\varphi_1,\ldots,\varphi_s) (1+o(1)),
\end{align}
Since $g(\varphi)=C\varphi(1+o(1))$,
$n\to\infty$, we get
\begin{equation}\label{I2}
I=\frac{C^{s(s-1)/2}\triangle(\xi_1,..,\xi_s)}{n^{s(s-1)/2}\prod_{j=0}^{s}j!}
\intd\limits_{\Omega_{n,s}}
\triangle^2(\varphi_1,..,\varphi_s)d\,\nu_{s}(\varphi_1,\ldots,\varphi_s)
(1+o(1)).
\end{equation}
Consider $T_\alpha$ of (\ref{int_okr1}) with $\alpha_1=\ldots=\alpha_s=+$, $\alpha_{s+1}=\ldots=\alpha_{2k}=-$.
Since the function
 $2c_{m,n}\kappa_4S_2((I-V_\alpha)\Lambda_0)
\prod\limits_{l=1}^{2k}\frac{\lambda_0^{-1/2}e^{i(\varphi_l+\alpha_l\varphi_0)}}
{1-\lambda_0^{-1/2}e^{i(\varphi_l+\alpha_l\varphi_0)}}$ is symmetric in $(\varphi_1,..,\varphi_s)$ and
$(\varphi_{s+1},..,\varphi_{2k})$, changing variables as $\sqrt{n}\varphi_j\to \varphi_j$ and
using formulas (\ref{I}) -- (\ref{I2}), and formula for the Selberg integral (see, e.g., \cite{Me:91}, Chapter 17),
we obtain
\begin{align}\notag
T_\alpha=&\dfrac{C_{0,s}(\hat{\xi})}{n^{(k-s)^2}\prod_{j=0}^{s}j!\prod_{l=0}^{2k-s}l!}
\intd\limits_{-\log n}^{\log n}\prod\limits_{j=1}^{2k}d\,\varphi_j\,
\triangle^2(\varphi_1,..,\varphi_s)\prod\limits_{j=1}^{s}
e^{-\frac{c_{+}\varphi_j^2}{2}}\\ \label{term_s}
&\times\triangle^2(\varphi_{s+1},..,\varphi_{2k})\prod\limits_{l=s+1}^{2k}e^{-\frac{c_{-}\varphi_l^2}{2}}
(1+o(1))
=\dfrac{C_{0,s}(\hat{\xi})(2\pi)^k}{c_+^{s^2/2}c_-^{(2k-s)^2/2}n^{(k-s)^2}}
(1+o(1)),
\end{align}
where $C_{0,s}(\hat{\xi})$ is an $n$-independent constant. This expression is of order $O(1)$ for $s=k$, and
it is of order $o(1)$ for $s\ne k$.
Hence, only the terms of (\ref{int_okr1}) with exactly $k$ of $\{\alpha_j\}_{j=1}^{2k}$ pluses contribute in
the limit (\ref{lim1}). If we take $s=k$ we obtain
\begin{align}\notag
C_{0,k}(\hat{\xi})=\lambda_0^{k(2k-1)/2}\bigg(\dfrac{e^{2i\varphi_0}}{\lambda_0^{1/2}\rho(\lambda_0)}\bigg)^{\frac{k(k-1)}{2}}
\bigg(\dfrac{e^{-2i\varphi_0}}{\lambda_0^{1/2}\rho(\lambda_0)}\bigg)^{\frac{k(k-1)}{2}}
(2i\sin\varphi_0)^{k^2}\\ \label{C_0}
\dfrac{\exp\{i\pi(\xi_{1}+...+\xi_{k}-\xi_{k+1}-...-\xi_{2k})\}}{\prod_{j=1}^k\prod_{l=k+1}^{2k}(\xi_j-\xi_{l})}
e^{k(k-1)\kappa_4(c-\lambda_0+1)^2c^{-1}}\\ \notag
=\dfrac{\lambda_0^{k^2}(2i\pi\rho(\lambda_0))^{k^2}}{(2\pi)^{2k}c^{k/2}}
e^{k(k-1)\kappa_4(c-\lambda_0+1)^2c^{-1}}\dfrac{e^{i\pi(\xi_{1}+...+
\xi_{k}-
\xi_{k+1}-...-\xi_{2k})}}{\prod_{j=1}^k\prod_{l=k+1}^{2k}(\xi_j-\xi_{l})}.
\end{align}
Hence, since it is easy to check that
\[
c_+c_-=\dfrac{4\pi^2\lambda_0^2\rho(\lambda_0)^2}{c_{m,n}},
\]
we get from (\ref{term_s}) and (\ref{C_0}) that $T_\alpha$ of (\ref{int_okr1}) with
$\alpha_1=...=\alpha_k=+$, $\alpha_{k+1}=...=\alpha_{2k}=-$ has the form
\begin{equation}\label{int_add1}
\dfrac{i^{k(k+1)}e^{i\pi(\xi_1+...+\xi_k-\xi_{k+1}-...-\xi_{2k})}}{(2i\pi)^k\prod\limits_{i,j=1}^k
(\xi_i-\xi_{k+j})}c^{k(k-1)/2}e^{k(k-1)\kappa_4(c-\lambda_0+1)^2c^{-1}}
\end{equation}
In view of the identity
\[
\dfrac{\mdet\left\{\dfrac{\sin(\pi(\xi_j-\xi_{k+l}))}{\pi(\xi_j-\xi_{k+l})}\right\}_{j,l=1}^k}{\Delta(\xi_1,..,\xi_k)
\Delta(\xi_{k+1},..,\xi_{2k})}=\dfrac{\mdet\left\{\dfrac{e^{i\pi(\xi_j-\xi_{k+l})}-e^{i\pi(\xi_{k+l}-\xi_j)}}
{\xi_j-\xi_{k+l}}\right\}_{j,l=1}^k}{
(2i\pi)^k\Delta(\xi_1,..,\xi_k)
\Delta(\xi_{k+1},..,\xi_{2k})}
\]
the determinant in the l.h.s. of (\ref{int_add1}) is a linear combination of $\exp\{i\pi\sum\limits_{j=1}^{2k}\alpha_j\xi_j\}$ over
the collection $\{\alpha_j\}_{j=1}^{2k}$, in which $m$ elements are pluses, and the rest are minuses.
By the virtue of the following formula (see \cite{Po-Se:76}, Problem 7.3)
\begin{equation}\label{iden}
(-1)^{\frac{k(k-1)}{2}}\dfrac{\prod_{j<l}(a_j-a_l)(b_j-b_l)}{\prod_{j,l=1}^k(a_j-b_l)}
=\det\left\{(a_j-b_l)^{-1}\right\}_{j,l=1}^m.
\end{equation}
the coefficient of $\exp\{i\pi(\xi_{k+1}+...+\xi_{2k}-
\xi_1-...-\xi_k)\}$ is
\[
\dfrac{\mdet\left\{(\xi_{k+l}-\xi_j)^{-1}\right\}_{j,l=1}^k}{
(2i\pi)^k\Delta(\xi_1,..,\xi_k)
\Delta(\xi_{k+1},..,\xi_{2k})}=\dfrac{(-1)^{\frac{k(k-1)}{2}}}{(-1)^{k^2}(2i\pi)^k\prod\limits_{i,j=1}^k(\xi_i-
\xi_{k+j})}.
\]
Other coefficients can be computed analogously. Thus, restricting the sum in (\ref{int_okr1})
to that over the collection $\{\alpha_j\}_{j=1}^{2k}$, in which exactly $k$ elements are pluses, and $k$ are minuses,
and using (\ref{int_add1}), we obtain Theorem \ref{thm:1} after a certain algebra.

\section{Asymptotic analysis at the edge of the spectrum.}

Let now $\lambda_0=\lambda_+$ (for $\lambda_0=\lambda_-$ the proof is similar) and
$\lambda_j=\lambda_++\xi_j/(n\gamma_+)^{2/3}$, $j=1,..,2k$, where $\lambda_+$ and $\gamma_+$ are defined in
(\ref{lam_pm}) and (\ref{rho_mp}), and $\xi_1,\ldots,\xi_{2k}\in [-M,M]\subset\mathbb{R}$.

According to (\ref{F_int}) we have
\begin{equation}\label{F_int_V_ed}
\dfrac{D_{2k}^{-1}(\lambda_+)}{(n\gamma_+)^{2k^2/3}}F_{2k}(\Lambda_{2k})=\widetilde{Z}_{2k}
\oint\limits_{\omega_0} \widetilde{W}_{n}(v_1,\ldots,v_{2k})\prod\limits_{j=1}^{2k}d\,v_j(1+o(1)),
\end{equation}
where $D_{2k}$ is defined in (\ref{D_2k}),
\begin{align}\label{W_n_ed}
\widetilde{W}_{n}(v_1,\ldots,v_{2k})&=
\exp\Big\{-n\sum\limits_{l=1}^{2k}V^{(+)}(v_l)+\sum\limits_{l=1}^{2k}\frac{n^{1/3}\xi_l}{\gamma^{2/3}} v_l+n(c_{m,n}-c)
\sum\limits_{l=1}^{2k}\log (1-v_l)\Big\}\\ \notag
&\times\dfrac{\triangle(V)}{\triangle(\widehat{\xi})}
\exp\Big\{2c_{m,n}\kappa_4S_2((I-V)\Lambda_0)
\prod\limits_{l=1}^{2k}\frac{v_l}{1-v_l}\Big\}\prod\limits_{j=1}^{2k}\dfrac{1}{v_j^{2k}},
\end{align}
\begin{equation}\label{V_ed}
V^{(+)}(v)=-\lambda_0v-c\log (1-v)+\log v-S_+,
\end{equation}
\begin{equation}\label{S_+}
S_+=-1-\sqrt{c}-c\log (1-(1+\sqrt{c})^{-1})-\log (1+\sqrt{c}),
\end{equation}
and
\begin{align}\label{Z_kp_ed}
\widetilde{Z}_{2k}&=L_{2k}e^{-n^{1/3}\alpha(\lambda_+)
\sum\limits_{j=1}^{2k}\xi_j
+2k(nc-m)\log(1-\lambda_+^{-1/2})},\\
L_{2k}&=\dfrac{n^{k(2k+1)/3}\gamma^{2k(k-1)/3} e^{-2k\kappa_4-2k(nc-m)\log(1-\lambda_+^{-1/2})}
}{2^{2k}\pi^{2k}c^{k/2}}.
\end{align}
We need the following lemma
\begin{lemma}\label{l:min_L_ed}
The function $\Re V^{(+)}(v)$ for $v=\lambda_+^{-1/2}e^{i\varphi},$ $\varphi\in (-\pi,\pi]$ attains its minimum at
\begin{equation}\label{v_pm_ed}
v_0:=\lambda_+^{-1/2}=(1+\sqrt c)^{-1}.
\end{equation}
Moreover, if $v\in\omega_0=\{v\in\mathbb{C}:v=\lambda_+^{-1/2}e^{i\varphi},\,\varphi\in (-\pi,\pi]\}$, $|v-v_0|\ge \delta$, where $\delta$ is small enough, then we have
for sufficiently big~$n$
\begin{equation}\label{ineqv_ReV_ed}
\Re V^{(+)}(v)\ge C\delta^4.
\end{equation}
\end{lemma}
\begin{proof}
Similarly to (\ref{ReV}) -- (\ref{ReV_pr}) we have
\begin{eqnarray}\label{ReV_pr_ed}
\dfrac{d}{d\,\varphi}\Re V^{(+)}(v_0)&=&\dfrac{d^2}{d\,\varphi^2}\Re V^{(+)}(v_0)
=\dfrac{d^3}{d\,\varphi^3}\Re V^{(+)}(v_0)=0,\\
\dfrac{d^4}{d\,\varphi^4}\Re V^{(+)}(v_0)&=&6.
\end{eqnarray}
Hence, $\varphi=0$ is a minimum point of the function $\Re V^{(+)}(\lambda_+^{-1/2}e^{i\varphi})$,
and $\Re V^{(+)}(\lambda_+^{-1/2}e^{i\varphi})$ is monotone increasing function for $\varphi\in [0,\pi)$
and monotone decreasing function for
 $\varphi\in (-\pi,0]$.

Expanding $\Re V^{(+)}(\lambda_+^{-1/2}e^{i\varphi})$ into the Taylor series
we obtain for $|\varphi|\le \delta$ similarly to (\ref{ineqv_ne0}):
\begin{equation}\label{ineqv_ne0_ed}
\Re V^{(+)}(\lambda_+^{-1/2}e^{i\varphi})=\varphi^4/4+O(\varphi^5).
\end{equation}
This and monotonicity of $\Re V^{(+)}(\lambda_+^{-1/2}e^{i\varphi})$ for $\varphi\ne 0$ imply
for $\varphi\in (-\pi,\pi]$, $|\varphi|\ge \delta$
\[
\Re V^{(+)}(\lambda_+^{-1/2}e^{i\varphi})\ge C\delta^4.
\]
Since $|v-v_0|=2\lambda_+^{-1/2}|\sin (\varphi/2)|\le \lambda_+^{-1/2} |\varphi|$, we get
(\ref{ineqv_ne0_ed}).
\end{proof}
Note that $|v_j|=\lambda_+^{-1/2}$, $j=1,..,2k$ and according to (\ref{cond_m}) $c_{m,n}-c=o(n^{-2/3})$,
$m,n\to\infty$. Since $\xi_1,\ldots,\xi_{2k}$ are distinct,
the inequality $|\triangle(T)/\triangle(\widehat{\xi})|\le C_1$ and (\ref{ineqv_ReV_ed}) yield
\begin{multline*}
\bigg|\widetilde{Z}_{2k}\int\limits_{\omega_0\setminus U_{\delta}(v_0)}\oint\limits_{\omega_0}..\oint\limits_{\omega_0}
\widetilde{W}_n(v_1,\ldots,v_{2k})\prod\limits_{j=1}^{2k}d\,v_j
\bigg|\\
\le C_1n^{k(2k+1)/3}e^{-C_2n(1+o(1))+C_3n^{1/3}},\,\,m,n\to\infty,
\end{multline*}
where $$U_\delta(v_0)=\{v\in\omega_0: |v-v_0|\le\delta\}.$$ Hence,
\begin{equation}\label{Int_okr1_ed}
\dfrac{D_{2k}^{-1}(\lambda_+)}{(n\gamma_+)^{2k^2/3}}F_{2k}(\Lambda_{2k})=\widetilde{Z}_{2k}
\oint\limits_{U_\delta(v_0)}
\widetilde{W}_{n}(v_1,\ldots,v_{2k})\prod\limits_{j=1}^{2k}d\,v_j(1+o(1))+O(e^{-Cn}).
\end{equation}
Since
\begin{equation*}
\dfrac{d}{d\,v} V^{(+)}(v_0)=\dfrac{d^2}{d\,v^2} V(v_0)=0,
\end{equation*}
we have for $|v-v_0|\le\delta$
\begin{equation}\label{V_okr}
V^{(+)}(v)=\gamma^{-2}(v-v_0)^3/3+O((v-v_0)^4), \quad |v-v_0|\to 0.
\end{equation}
Thus, we can write for $v$ satisfying $|v-v_0|\le\delta$
\begin{equation}\label{v_sm}
V^{(+)}(v)=\gamma^{-2}\chi^3(v)/3,
\end{equation}
where $\chi(v)$ is analytic in the $\delta$-neighborhood of $v_0$ with the analytic
inverse $z(\varphi)$ (we choose $\chi(v)$ such that $\chi(v)\in\mathbb{R}$ for $v\in\mathbb{R}$).

Changing variables to $v_j=z(\varphi_j)$, $j=1,..,2k$, we rewrite (\ref{Int_okr1_ed}) as
\begin{align}\notag
&\dfrac{D_{2k}^{-1}(\lambda_+)}{(n\gamma_+)^{2k^2/3}}F_{2k}(\Lambda_{2k})=L_{2k}
\oint\limits_{\widetilde{U}_{\delta,\varphi}}
e^{-n\gamma^{-2}\sum\limits_{l=1}^{2k}\varphi_l^3/3+\sum\limits_{l=1}^{2k}\frac{n^{1/3}\xi_l}{\gamma^{2/3}}
 (z(\varphi_l)-v_0)+n(c_{m,n}-c)\sum\limits_{l=1}^{2k}\log\frac{1-z(\varphi_l)}{1-v_0}}
\\ \label{Int_okr2_ed}
 &\quad \times
e^{2c_{m,n}\kappa_4S_2((I-Z)\Lambda_0)
\prod\limits_{l=1}^{2k}\frac{z(\varphi_l)}{1-z(\varphi_l)}}\dfrac{\triangle(Z)}
{\triangle(\widehat{\xi})}
\prod\limits_{j=1}^{2k}\frac{z^\prime(\varphi_j)}{z^{2k}(\varphi_j)}d\,\varphi_j(1+o(1))+O(e^{-Cn})\\ \notag
 &\quad =: L_{2k}\intd\limits_{\widetilde{U}_{\delta,\varphi}} \widehat{W}(\varphi_1,\ldots,\varphi_{2k})\prod\limits_{j=1}^{2k}
d\,\varphi_j(1+o(1))+O(e^{-Cn}),
\end{align}
where $L_{2k}$ is defined in (\ref{Z_kp_ed}),
\begin{equation}\label{Z}
Z=\hbox{diag}\,\{z(\varphi_1),\ldots, z(\varphi_{2k})\},
\end{equation}
\begin{equation}\label{U_del,phi}
\widetilde{U}_{\delta,\varphi}=\{\varphi\in \mathbb{C}: z(\varphi)\in U_\delta(v_0)\}.
\end{equation}
Moreover, we have from (\ref{v_sm})
\begin{equation}\label{chi_v_0}
\chi(v_0)=0,\quad \dfrac{d}{d\,z}\chi(v_0)=1,
\end{equation}
hence
\begin{equation}\label{ogr_pr_chi}
0<C_1<|\chi^\prime(v)|<C_2, \quad |v-v_0|\le\delta.
\end{equation}
If $\widetilde{\sigma}=\{z\in\mathbb{C}:|z-z^*_{0,n}|\le \delta\}$, then $\chi(\partial\widetilde{\sigma})$ is a closed
curve encircling $\varphi=0$ and lying between the circles
$\sigma_1=\{\varphi\in\mathbb{C}:|\varphi|=C_1\delta\}$ and
$\sigma_2=\{\varphi\in\mathbb{C}:|\varphi|=C_2\delta\}$ for $0<C_1<C_2$.
We have from (\ref{chi_v_0})
\begin{equation}\label{z_pr}
z(0)=v_0,\quad z^\prime(0)=1, \quad 0<C_1<|z^\prime(\varphi)|<C_2,\quad \varphi\in \chi(\widetilde{\sigma}).
\end{equation}
According to Lemma \ref{l:min_L_ed}, $\Re V^{(+)}(v)\ge 0$ for $v\in U_\delta(v_0)$ and we get $\Re \varphi^3_j\ge
0$ for $\varphi_j\in \widetilde{U}_{\delta,\varphi}$, i.e.,
\[
\cos (3\arg\varphi_j)\ge 0, \quad \varphi_j\in \widetilde{U}_{\delta,\varphi},
\]
where $\widetilde{U}_{\delta,\varphi}$ is defined in (\ref{U_del,phi}). Hence, $\widetilde{U}_{\delta,\varphi}$ can be located only in
the sectors
\[
-\pi/6\le \arg\varphi\le\pi/6, \quad
\pi/2\le \arg\varphi\le 5\pi/6,\quad
7\pi/6\le \arg\varphi\le 3\pi/2.
\]
Besides, $\chi$ is conformal in $\widetilde{\sigma}$ (see (\ref{ogr_pr_chi})), hence angle-preserving.
Taking into account that $\chi(v)\in\mathbb{R}$ for $v\in\mathbb{R}$, the angle between $\omega_0$
and the real axis at the point $v_0$ is $\pi/2$, and that $\widetilde{U}_{\delta,\varphi}$ is a continuous curve, we obtain that
$\widetilde{U}_{\delta,\varphi}$ can be located only in the sectors
\begin{equation}\label{sectors}
\pi/2\le \arg\varphi\le 5\pi/6,\quad
7\pi/6\le \arg\varphi\le 3\pi/2.
\end{equation}
Note that we can take any curve $\widetilde{U}(\varphi)$ instead of $\widetilde{U}_{\delta,\varphi}$ provided that
$\widetilde{U}(\varphi)$ and $\omega_0\setminus U_\delta(v_0)$ are "glued", i.e., the union of $z(\widetilde{U}(\varphi))$ and
$\omega_0\setminus U_\delta(v_0)$ form a closed contour encircling $0$.
Let us take
\begin{multline*}
\widetilde{U}(\varphi)=\{\varphi\in\mathbb{C}:\arg \varphi=2\pi/3,\,\varphi\in \chi(\widetilde{\sigma})\}\\
\cup\{\varphi\in\mathbb{C}:\arg \varphi=4\pi/3,\,\varphi\in \chi(\widetilde{\sigma})\}
\cup U_{1,\delta}\cup U_{2,\delta},
\end{multline*}
where $\widetilde{\sigma}=\{v\in\mathbb{C}:|v-v_0|\le\delta\}$, $U_{1,\delta}$ is a curve along $\chi(\partial \widetilde{\sigma})$
from the point of intersection
of the ray $\arg\varphi=2\pi/3$ and $\chi(\partial \widetilde{\sigma})$ to the point $\varphi_{1,\delta}$ of intersection
of  $\widetilde{U}_{\delta,\varphi}$ and $\chi(\partial \widetilde{\sigma})$ ($\pi/2<\arg\varphi_{1,\delta}<5\pi/6$),
and $U_{2,\delta}$ is a curve along $\chi(\partial \widetilde{\sigma})$ from the point of intersection
of the ray $\arg\varphi=4\pi/3$ and $\chi(\partial \widetilde{\sigma})$ to the point $\varphi_{2,\delta}$ of intersection
of $\widetilde{U}_{\delta,\varphi}$ and $\chi(\partial \widetilde{\sigma})$ ($7\pi/6<\arg\varphi_{2,\delta}<3\pi/2$).
%
%
According to Lemma \ref{l:min_L_ed} and (\ref{v_sm}), $\Re \varphi^3_{1,\delta}=r^3\cos 3\varphi_0>C>0$, where
$r=|\varphi_{1,\delta}|$, $\varphi_0=\arg\varphi_{1,\delta}$. Since $0<C_1<r<C_2$, we have
\[
\cos 3\varphi_0\ge C/C_2^3>0.
\]
Moreover, it is easy to see that $\cos (3\arg\varphi_1)>\cos 3\varphi_0$ along $U_{1,\delta}$ (since
$\cos 3x$ is monotone increasing for $x\in[\pi/2,2\pi/3]$ and monotone decreasing for
$x\in [2\pi/3,5\pi/6]$).
This and $|\varphi_j|>C_1$ imply for $\varphi_j\in L_{1,\delta}$
\begin{equation}\label{ots_re}
\Re \left(\dfrac{\gamma^{-2}\varphi_j^3}{3}\right)>C>0,\quad \varphi_1\in U_{1,\delta}.
\end{equation}
Also we have from (\ref{z_pr})
\[
|z(\varphi_j)-v_0|\le C_2|\varphi_j|<C,\quad\varphi_j\in \chi(\widetilde{\sigma}).
\]
This, (\ref{ots_re}), $c_{m,n}-c=o(n^{-2/3})$, $m,n\to\infty$, and  (\ref{z_pr}) yield
\begin{equation}\label{ots_L_d}
\Big| \widehat{W}(\varphi_1,\ldots,\varphi_{2k})\Big|\le e^{-Cn+o(n)},\quad \varphi_1\in U_{1,\delta},\quad
\varphi_j\in \widetilde{U}(\varphi_j),\,\,j>1.
\end{equation}
Hence, the integral over $U_{1,\delta}$ does not contribute to the l.h.s. of (\ref{Int_okr2_ed})
The same statement we can prove for $U_{2,\delta}$.
Thus, we have shown that integral over $\widetilde{U}_{\delta,\varphi}$ in (\ref{Int_okr2_ed}) can be replaced
to the integral over the contour
\begin{equation}\label{l}
\widetilde{l}=\{\varphi\in\mathbb{C}:\arg \varphi=2\pi/3,\,\varphi\in \chi(\widetilde{\sigma})\}
\cup\{\varphi\in\mathbb{C}:\arg \varphi=4\pi/3,\,\varphi\in \chi(\widetilde{\sigma})\}.
\end{equation}
According to the choice of $\widetilde{l}$, we have
\begin{equation}\label{phi^3}
\Re \varphi_j^3=r_j^3,\quad \varphi_j\in \widetilde{l},
\end{equation}
where $r_j=|\varphi_j|$.

Set now $$\sigma_n=\{\varphi\in\mathbb{C}:|\varphi|\le \varepsilon_n^{-1/2}n^{-1/3}\},$$
where $\varepsilon_n$ is defined in (\ref{cond_m}).
Note that we assume that $\varepsilon_n^{-1/2}n^{-1/3}\to 0$, $n\to\infty$.
In other case we can take
$\sigma_n=\{\varphi\in\mathbb{C}:|\varphi|\le \log n/n^{1/3}\}$ and the proof will be similarly.

 It is easy to see that
 $\sigma_n\subset \chi(\widetilde{\sigma})$ for sufficiently big $n$. Besides,
we have from (\ref{z_pr}) for $\varphi\in\sigma_n$
\begin{equation}\label{razl_z}
\begin{array}{c}
z(\varphi)=v_0+\varphi+O(\varepsilon_n^{-1} n^{-2/3}),\,\,n\to\infty,\\
z^\prime(\varphi)=1+O(\varepsilon_n^{-1/2}n^{-1/3}),\,\,n\to\infty.
\end{array}
\end{equation}
 Taking into account (\ref{cond_m}), (\ref{razl_z}), (\ref{phi^3}), and
\[
\bigg|\log \dfrac{1-z(\varphi)}{1-v_0}\bigg|\le \Big|\dfrac{v_0-z(\varphi)}{1-v_0}\Big|,
\]
 we obtain for $\varphi_1\in \widetilde{l}\setminus \sigma_n$, $\varphi_j\in \widetilde{l}$, $j=2,..,2k$
\begin{equation}\label{ots_l_n}
\Big|\widehat{W}(\varphi_1,\ldots,\varphi_{2k})\Big|\le C_1e^{-C_2nr_1^3+C_3n^{1/3}r_1}
\end{equation}
where $r_1=|\varphi_1|\ge \varepsilon_n^{-1/2} n^{-1/3}$. Since $n^{1/3}r_1\ge \varepsilon_n^{-1/2}$ for $\varphi_1\in
\widetilde{l}\setminus \sigma_n$, the integral over $\widetilde{l}\setminus \sigma_n$ is $O(e^{-C\varepsilon_n^{-3/2}})$ as $n\to \infty$.
Hence,
\begin{equation}\label{Int_okr3_ed}
\dfrac{D_{2k}^{-1}(\lambda_+)}{(n\gamma_+)^{2k^2/3}}F_{2k}(\Lambda_{2k})=L_{2k}(1+o(1))
\int\limits_{\widetilde{l}\bigcap \sigma_n}
\widehat{W}(\varphi_1,\ldots,\varphi_{2k})
\prod\limits_{j=1}^{2k}d\,\varphi_j+O(e^{-C\varepsilon_n^{-3/2}}),
\end{equation}
where $L_{2k}$ and $\widehat{W}(\varphi_1,\ldots,\varphi_{2k})$ are defined in (\ref{Z_kp_ed}) and (\ref{Int_okr2_ed}).
This, (\ref{S}), and (\ref{razl_z}) imply
\begin{equation}\label{Int_okr4_ed}
\begin{array}{c}
\dfrac{D_{2k}^{-1}}{(n\gamma_+)^{2k^2/3}}F_{2k}(\Lambda_{2k})=
\frac{L_{2k}e^{2k(2k-1)\kappa_4}}{v_0^{4k^2}}
\int\limits_{\widetilde{l}\bigcap \sigma_n}
e^{-n\gamma^{-2}\sum\limits_{l=1}^{2k}\varphi_l^3/3+\sum\limits_{l=1}^{2k}\frac{n^{1/3}\xi_l}{\gamma^{2/3}}\varphi_l}
\dfrac{\triangle(\Phi)}
{\triangle(\widehat{\xi})}\\
\times (1+\delta(\varphi_1,\ldots,\varphi_{2k}))\prod\limits_{j=1}^{2k}
d\,\varphi_j+O(e^{-C\varepsilon_n^{-3/2}}),
\end{array}
\end{equation}
where $\delta(\varphi_1,\ldots,\varphi_{2k})$ collects the reminder terms which appear when we replace
$z(\varphi_j)\to v_0+\varphi_j+O(\varphi_j^2)$, $z^\prime(\varphi_j)\to 1+O(\varphi_j)$ and
$\log \frac{1-z(\varphi_j)}{1-v_0}\to\frac{\varphi_j}{1-v_0}+O(\varphi_j^2)$, $j=1,..,2k$. Hence
\begin{equation*}
|\delta(\varphi_1,\ldots,\varphi_{2k})|\le C(|\varphi_1|+\ldots+|\varphi_{2k}|).
\end{equation*}
Changing variables in (\ref{Int_okr4_ed}) as $\gamma^{-2/3} n^{1/3}\varphi_j\to i\varphi_j$ we obtain
in new variables
\begin{equation*}\label{delta}
|\widetilde{\delta}(\varphi_1,\ldots,\varphi_{2k})|=|\delta(i\gamma^{2/3} n^{-1/3}\varphi_1,\ldots,
 i\gamma^{2/3} n^{-1/3}\varphi_{2k})|\le Cn^{-1/3}(|\varphi_1|+\ldots+|\varphi_{2k}|).
\end{equation*}
Therefore, using (\ref{gam}), (\ref{Z_kp_ed}), and (\ref{v_pm_ed}), we obtain
\begin{align}\notag
\dfrac{D_{2k}^{-1}(\lambda_+)}{(n\gamma_+)^{2k^2/3}}F_{2k}(\Lambda_{2k})&=\frac{L_{2k}e^{2k(2k-1)\kappa_4}(1+o(n^{-1/3}))}
{v_0^{4k^2}(-i\gamma^{-2/3}n^{1/3})^{k(2k+1)}}\\ \notag
&\times\intd\limits_{S}
e^{i\sum\limits_{l=1}^{2k}\varphi_l^3/3+\sum\limits_{l=1}^{2k}i\xi_l\varphi_l}
\dfrac{\triangle(\Phi)}
{\triangle(\widehat{\xi})}
\prod\limits_{j=1}^{2k}d\,\varphi_j+O(e^{-C\varepsilon_n^{-3/2}})\\ \label{Int_okr5_ed}
&=e^{4k(k-1)\kappa_4}(-1)^{k^2}c^{k(k-1)/2}(1+o(n^{-1/3}))\\ \notag
&\times \dfrac{i^k}{(2\pi)^{2k}} \intd\limits_{S}
e^{i\sum\limits_{l=1}^{2k}\varphi_l^3/3+\sum\limits_{l=1}^{2k}i\xi_l\varphi_l}
\dfrac{\triangle(\Phi)}
{\triangle(\widehat{\xi})}
\prod\limits_{j=1}^{2k}d\,\varphi_j+O(e^{-C\varepsilon_n^{-3/2}})
,
\end{align}
where $S$ is defined in (\ref{Ai}).

Consider
\begin{multline}\label{K}
K(\widehat{\xi}):=\dfrac{i^{k}}{(2\pi)^{2k}}\intd\limits_{S}
e^{i\sum\limits_{l=1}^{2k}\varphi_l^3/3+\sum\limits_{l=1}^{2k}i\xi_l\varphi_l}
\triangle(\Phi)\prod\limits_{j=1}^{2k}d\,\varphi_j\\
=\dfrac{i^{k}}{(2\pi)^{2k}}\intd \det\bigg\{\varphi_l^{j-1}e^{i\varphi_l^3/3+i\xi_l\varphi_l}
\bigg\}_{j,l=1}^{2k}\prod\limits_{j=1}^{2k}d\,\varphi_j.
\end{multline}
Integrating by parts, we have for $j\ge 3$
\begin{multline*}
i\intd\limits_{S}\varphi_l^{j-1}e^{i\varphi_l^3/3+i\xi_l\varphi_l}d\,\varphi_l=
\intd\limits_{S}\varphi_l^{j-3}e^{i\xi_l\varphi_l} \dfrac{d}{d\varphi_l}
e^{i\varphi_l^3/3}d\,\varphi_l\\
=-\intd\limits_{S}((j-3)\varphi_l^{j-4}+i\xi_l\varphi_l^{j-3})e^{i\varphi_l^3/3+i\xi_l\varphi_l}d\,\varphi_l.
\end{multline*}
Applying this identity to each line, starting from third, we observe that the first term in the r.h.s. gives zero
contribution. Repeating this procedure and replacing  and rearranging the lines, we obtain from (\ref{K})
\begin{equation}\label{K1}
K(\widehat{\xi})=\dfrac{i^{k}}{(2\pi)^{2k}}\intd (-1)^{k(k-1)/2}\det\bigg\{\varphi_l^{q_j}\xi_l^{r_j}e^{i\varphi_l^3/3+i\xi_l\varphi_l}
\bigg\}_{j,l=1}^{2k}\prod\limits_{j=1}^{2k}d\,\varphi_j,
\end{equation}
where $j=q_jk+r_j$, $q_j=0,1$, $r_j=0,1,..,k-1$. Thus,
\begin{equation}\label{K2}
\begin{array}{c}
\dfrac{K(\widehat{\xi})}{\Delta(\widehat{\xi})}=\dfrac{i^{k}}{(2\pi)^{2k}}
\intd\limits_{S} \sum\limits_{j_1< ...< j_k}\dfrac{\triangle(\xi_{j_1},..,\xi_{j_k})\overline{\triangle}
(\xi_{j_1},..,\xi_{j_k})\prod\limits_{s=1}^k\varphi_{j_s}}
{(-1)^{j_1+...+j_k}\Delta(\xi_1,..,\xi_{2k})}e^{i\sum\limits_{l=1}^{2k}\varphi_l^3/3+\sum\limits_{l=1}^{2k}i\xi_l\varphi_l}
\prod\limits_{j=1}^{2k}d\,\varphi_j\\
=\dfrac{i^{k}}{(2\pi)^{2k}}\intd\limits_{S} \sum\limits_{j_1< ...< j_k}\dfrac{(-1)^{k(k+1)/2}\prod\limits_{s=1}^k\varphi_{j_s}}
{\prod\limits_{s=1}^k\prod\limits_{t\ne j_1,..,j_k}(\varphi_{j_s}-\varphi_t)\,\hbox{sign}\,(t-j_s)}e^{i\sum\limits_{l=1}^{2k}\varphi_l^3/3+\sum\limits_{l=1}^{2k}i\xi_l\varphi_l}
\prod\limits_{j=1}^{2k}d\,\varphi_j,
\end{array}
\end{equation}
where the sum is over all collections $1\le j_1< \ldots < j_k\le 2k$ and $\overline{\triangle}
(\xi_{j_1},..,\xi_{j_k})$ is the Vandermonde determinant of $\{\xi_j\}$ with $j\ne j_1,..,j_k.$
 Consider
\begin{equation}\label{ker}
\dfrac{(-1)^{k^2}}{(2\pi)^{2k}}\intd\limits_{S}\frac{\det\bigg\{\dfrac{i\varphi_j-i\varphi_{l+k}}{\xi_j-\xi_{l+k}}\bigg\}_{j,l=1}^k }
{\Delta(\xi_1,..,\xi_{k})
\Delta(\xi_{k+1},..,\xi_{2k})}
e^{i\sum\limits_{l=1}^{2k}\varphi_l^3/3+\sum\limits_{l=1}^{2k}i\xi_l\varphi_l}\prod\limits_{j=1}^{2k}d\,\varphi_j.
\end{equation}
According to the identity (\ref{iden}) the coefficient of $\prod\limits_{s=1}^k\varphi_{j_s}$ in (\ref{ker}) is
\[
\dfrac{i^{k}}{(2\pi)^{2k}}\cdot\dfrac{(-1)^{k(k+1)/2}}
{\prod\limits_{s=1}^k\prod\limits_{t\ne j_1,..,j_k}(\varphi_{j_s}-\varphi_t)\,\hbox{sign}\,(t-j_s)}.
\]
Thus, $K(\widehat{\xi})/\Delta(\widehat{\xi})$ is equal to (\ref{ker}), and (\ref{Int_okr5_ed})
yield the assertion of Theorem \ref{thm:2}.
%

\end{document}